\documentclass[12pt]{article}
\usepackage[latin9]{inputenc}
\usepackage{float}
\usepackage{mathtools}
\usepackage{amsmath}
\usepackage{amsthm}
\usepackage{amssymb}
\usepackage{graphicx}
\usepackage[a4paper]{geometry}
\usepackage{setspace}
\usepackage[authoryear]{natbib}
\makeatletter

\newcommand*\LyXZeroWidthSpace{\hspace{0pt}}
\providecommand{\tabularnewline}{\\}

\theoremstyle{plain}
\newtheorem{lem}{\protect\lemmaname}
\theoremstyle{plain}
\newtheorem{thm}{\protect\theoremname}

\renewcommand{\hat}{\widehat}
\newtheorem{assumption}{Assumption}

\providecommand{\theoremname}{Theorem}

\makeatother

\providecommand{\lemmaname}{Lemma}
\providecommand{\theoremname}{Theorem}

\begin{document}

\title{The Robust F-Statistic as a Test for Weak Instruments}
\author{Frank Windmeijer\thanks{Email address: frank.windmeijer@stats.ox.ac.uk}
\\
 {\small Department of Statistics and Nuffield College, University
of Oxford}\\
{\small New Road, Oxford OX1 1NF, UK}\\
}
\maketitle
\begin{abstract}
\noindent\baselineskip=15pt For the linear model with a single endogenous
variable, \citet{OleaPflueger2013} proposed the effective F-statistic
as a test for weak instruments in terms of the Nagar bias of the two-stage
least squares (2SLS) or limited information maximum likelihood (LIML)
estimator relative to a benchmark worst-case bias. We show that their
methodology for the 2SLS estimator applies to a class of linear generalized
method of moments (GMM) estimators with an associated class of generalized
effective F-statistics. The standard robust F-statistic is a member
of this class. The associated GMMf estimator, with the extension ``f''
for first-stage, has the weight matrix based on the first-stage residuals.
In the grouped-data IV designs of \citet{IAndrewsREStat2018} with
moderate and high levels of endogeneity and where the robust F-statistic
is large but the effective F-statistic is small, the GMMf estimator
is shown to behave much better in terms of bias than the 2SLS estimator.
\end{abstract}
\noindent{\small\textbf{Keywords:}}{\small{} Instrumental variables,
weak instruments, heteroskedasticity, robust F-statistic, GMM}{\small\par}

\noindent{\small\textbf{JEL Codes:}}{\small{} C12, C26}{\small\par}

\thispagestyle{empty}

\baselineskip=20pt

\pagebreak{}

\pagenumbering{arabic} \setcounter{page}{1}

\section{Introduction}

\label{Sec:Intro}It is commonplace to report the first-stage F-statistic
to test for weak instruments in linear models with a single endogenous
variable, estimated by two-stage least squares (2SLS). This follows
the work of \citet{StaigerStock1997} and \citet{StockYogo2005},
with the latter providing critical values for the first-stage non-robust
F-statistic for null hypotheses of weak instruments in terms of bias
of the 2SLS estimator relative to that of the OLS estimator and Wald-test
size distortions. These weak-instruments critical values for the non-robust
F-statistic are valid only under homoskedasticity (i.e.\ conditional
homoskedasticity, no serial correlation and no clustering) of both
the first-stage and structural errors, and do not apply to the robust
F-statistic in general designs, see \citet{BunHaan2010}, \citet{OleaPflueger2013}
and \citet{IAndrewsREStat2018}. In particular, \citet{IAndrewsREStat2018}
found for some cross-sectional heteroskedastic designs that the standard
2SLS confidence intervals had large coverage distortions even for
very large values of the robust F-statistic. For example, he found
for a high endogeneity design that ``the 2SLS confidence set has
a 15\% coverage distortion even when the mean of the first-stage robust
F-statistic is 100,000'', \citet[Supplementary Appendix, p 11]{IAndrewsREStat2018}.

For general heteroskedasticity, which we mean to cover the cases of
conditional heteroskedasticity, serial correlation and clustering,
\citet{OleaPflueger2013} proposed the first-stage effective F-statistic
and derived critical values for the null of weak instruments in terms
of the Nagar bias of the 2SLS or LIML estimator, relative to an estimator-specific
benchmark worst-case bias. We focus here on their results for the
2SLS estimator. As shown in Section \ref{subsec:MC}, the effective
F-statistics in the designs of \citet{IAndrewsREStat2018} do not
reject the null of weak instruments. In their review paper \citet*[p 729]{AndrewsStockSun2019}
recommend ``that researchers judge instrument strength based on the
effective F-statistic of \citet{OleaPflueger2013}''.

The effective F-statistic is specific to the Nagar bias of the 2SLS
or LIML estimator and the main contribution of this paper is that
the 2SLS related methods of \citet{OleaPflueger2013} apply to a wider
class of linear generalized method of moments (GMM) estimators resulting
in a class of associated generalized effective F-statistics. The robust
F-statistic is a member of this class, and we call its associated
GMM estimator the GMMf estimator, with the extension ``f'' for first-stage.
This is because the weight matrix of the GMMf estimator is based on
the first-stage residuals, with $k_{z}$ times the robust F-statistic
being the denominator of the GMMf estimator, where $k_{z}$ is the
number of excluded instruments. This is similar to the relationship
of the non-robust F-statistic and the 2SLS estimator.

In practice, reported 2SLS estimation results with robust standard
errors are often accompanied by the robust first-stage F-statistic,
as most statistical packages automatically provide these together.
Whilst the robust F-statistic can be used as a test for underidentification,
it can not be used as a test for weak instruments related to the performance
of the 2SLS estimator. Hence, for the single-endogenous variable case,
robust 2SLS estimation results should be accompanied by the effective
F-statistic and its critical value, as per the advice of \citet{AndrewsStockSun2019}.
The critical value of the robust F-statistic in relation to the test
for weak instruments for the GMMf estimator could also be computed.
If the situation is such, like in the \citet{IAndrewsREStat2018}
example, that the effective F-statistic is small and indicates weak
instruments problems for the 2SLS estimator, but the robust F-statistic
is large, rejecting the null of weak instruments for the GMMf estimator,
then the latter could be preferred and reported.\footnote{An extension of the ``weakivtest'' command of \citet{weakivtest}
in Stata, \citet{Stata}, called ``gfweakivtest'' for calculating
these is available from https://github.com/stat0357/gfweakivtest.} This is illustrated in Section \ref{subsec:MC}, where we replicate
the Monte Carlo analysis of \citet{IAndrewsREStat2018}. The design
is the same as a grouped-data one, see \citet{AngristJoE1991} and
the discussion in \citet{AngristPischke2009}, where the instruments
are mutually exclusive group membership indicators. In the two designs
considered, there is in each only one informative group, but the first-stage
heteroskedasticity is such that the 2SLS estimator does not utilize
this information well, whereas the GMMf estimator gives almost all
the weight to the informative groups.

Section \ref{sec:Model} introduces the single-endogenous variable
linear model specification, main assumptions, effective and robust
F-statistics and the GMMf estimator. Section \ref{sec:Fgeff} then
formulates the class of generalized F-statistics for the class of
linear GMM estimators and shows that the weak-instruments testing
methods developed by \citet{OleaPflueger2013} apply straightforwardly
to this class. The section then provides a summary of the \citet{OleaPflueger2013}
results. Section \ref{sec:RobF} shows how the general results and
specifications simplify for the robust-F statistic in relation to
the Nagar bias of the GMMf estimator. As the Nagar bias is relative
to a benchmark worse-case bias, which is estimator specific, we harmonize
in Section \ref{sec:Harmon} the benchmark bias and propose the use
of the worst-case OLS bias as the benchmark, which applies to the
class of GMM estimators considered.

The weak-instruments test procedures considered here only apply to
linear models with a single endogenous regressor. \citet{LewisMertens2022}
develop an extension of the \citet{OleaPflueger2013} method to the
multiple endogenous variable case for the 2SLS estimator, but they
do not consider such an extension for the wider class of GMM estimators.

\section{Model, Assumptions and F-Statistics}

\label{sec:Model}

We have a sample $\left\{ y_{i},x_{i},z_{i}\right\} _{i=1}^{n}$,
where $z_{i}$ is a $k_{z}$-vector of instrumental variables. We
are interested in the effect of $x$ on $y$ in a linear model specification,
where $x$ is endogenously determined. We consider the linear structural
and first-stage specifications
\begin{align}
y & =x\beta+u\label{eq:struc}\\
x & =Z\pi+v_{2},\label{eq:fstage}
\end{align}
where $y$, $x$, $u$ and $v_{2}$ are $n$-vectors and $Z$ an $n\times k_{z}$
matrix. Other exogenous explanatory variables, including the constant
have been partialled out. The reduced-form specification for $y$
is then given by
\begin{equation}
y=Z\pi\beta+v_{1}=Z\pi_{y}+v_{1},\label{eq:reduced}
\end{equation}
where $v_{1}=u+\beta v_{2}$, and $\pi_{y}=\pi\beta$.

Following \citet{OleaPflueger2013} (henceforth MOP), we make the
following assumptions.

\begin{assumption}\LyXZeroWidthSpace{}

\label{Ass:all}
\begin{enumerate}
\item \label{Ass:WIA} Weak-instruments asymptotics. The vector $\pi$ is
local to zero,
\[
\pi=\pi_{n}=c/\sqrt{n},
\]
where $c$ is a fixed vector $c\in\mathbb{R}^{k_{z}}$.
\item \label{Ass:lim} As $n\rightarrow\infty$,
\begin{align*}
\frac{1}{n}Z'Z & \overset{p}{\rightarrow}Q_{zz};\\
\\\frac{1}{n}\left[v_{1}\,v_{2}\right]'\left[v_{1}\,v_{2}\right] & \stackrel{p}{\rightarrow}\Sigma_{v};\\
\\\frac{1}{\sqrt{n}}\left(\begin{array}{c}
Z'v_{1}\\
Z'v_{2}
\end{array}\right) & \overset{d}{\rightarrow}\left(\begin{array}{c}
\psi_{1}\\
\psi_{2}
\end{array}\right)\sim N\left(0,W\right),
\end{align*}
 with $Q_{zz}$, $\Sigma_{v}$ and $W$ finite, positive definite
matrices, and
\begin{align*}
\Sigma_{v} & =\left[\begin{array}{cc}
\sigma_{1}^{2} & \sigma_{12}\\
\sigma_{12} & \sigma_{2}^{2}
\end{array}\right];\\
\\W & =\left[\begin{array}{cc}
W_{1} & W_{12}\\
W_{12}' & W_{2}
\end{array}\right].
\end{align*}
 
\item There exists a sequence of positive definite estimates $\left\{ \hat{W}_{n}\right\} $,
such that $\hat{W}_{n}\overset{p}{\rightarrow}W$ as $n\rightarrow\infty$.
\end{enumerate}
\end{assumption}

In the remainder, we drop the subscript $n$ from $\hat{W}_{n}$ for
ease of exposition, and, commensurate with the partitioning of $W$,
\[
\hat{W}=\left[\begin{array}{cc}
\hat{W}_{1} & \hat{W}_{12}\\
\hat{W}_{12}' & \hat{W}_{2}
\end{array}\right].
\]
The two-stage least squares (2SLS) estimator is given by
\[
\hat{\beta}_{2sls}=\frac{x'P_{Z}y}{x'P_{Z}x},
\]
where $P_{Z}=Z\left(Z'Z\right)^{-1}Z'$. The standard non-robust first-stage
F-statistic is
\[
\hat{F}=\frac{x'P_{Z}x}{k_{z}\hat{\sigma}_{v_{2}}^{2}},
\]
where $\hat{\sigma}_{v_{2}}^{2}=\hat{v}_{2}'\hat{v}_{2}/n$, $\hat{v}_{2}=\left(I_{n}-P_{Z}\right)x$,
where $I_{n}$ is the identity matrix of order $n$. Note that we
refrain throughout from finite sample degrees-of-freedom corrections
in the exposition. It follows that we can alternatively express the
2SLS estimator as
\begin{equation}
\hat{\beta}_{2sls}=\frac{\hat{\pi}'\left(V\hat{a}r\left(\hat{\pi}\right)\right)^{-1}\hat{\pi}_{y}}{\hat{\pi}'\left(V\hat{a}r\left(\hat{\pi}\right)\right)^{-1}\hat{\pi}}=\frac{\hat{\pi}'\left(V\hat{a}r\left(\hat{\pi}\right)\right)^{-1}\hat{\pi}_{y}}{k_{z}\hat{F}},\label{eq:2slsfr}
\end{equation}
where $\hat{\pi}$ and $\hat{\pi}_{y}$ are the OLS estimators of
$\pi$ and $\pi_{y}$ in the first-stage and reduced-form models (\ref{eq:fstage})
and (\ref{eq:reduced}), and $V\hat{a}r\left(\hat{\pi}\right)=\hat{\sigma}_{v_{2}}^{2}\left(Z'Z\right)^{-1}$,
the non-robust estimator of the variance of $\hat{\pi}$.

$\hat{F}$ can be used as a test for underidentification and as a
test for weak instruments in terms of the bias of the 2SLS estimator
relative to that of the OLS estimator of $\beta$, or the size distortion
of the Wald test for hypotheses on $\beta$, \citet{StockYogo2005}.
The Stock and Yogo critical values are valid only under conditional
homoskedasticity of both $u$ and $v_{2}$, i.e.\ $\mathbb{E}\left[u_{i}|z_{i}\right]=\sigma_{u}^{2}$
and $\mathbb{E}\left[v_{2i}|z_{i}\right]=\sigma_{v_{2}}^{2}$, or
$W=\Sigma_{v}\otimes Q_{zz}$, with $\Sigma_{v}=\mathbb{E}\left[v_{i}v_{i}'\right]$
and $v_{i}=\left(v_{1i},v_{2i}\right)'$.

The robust F-statistic is given by
\begin{equation}
\hat{F}_{r}=\frac{x'Z\hat{W}_{2}^{-1}Z'x}{nk_{z}}\label{eq:Fr}
\end{equation}
and is a standard test statistic for testing $H_{0}:\pi=0$ under
general forms of heteroskedasticity. But $\hat{F}_{r}$ cannot be
used as a test for weak instruments in relation to the behaviour of
the 2SLS estimator, see the discussion in \citet[pp 738-739]{AndrewsStockSun2019},
summarized at the end of this section. \citet{IAndrewsREStat2018}
showed in a grouped-data IV design that $\hat{F}_{r}$ could take
very large values, of the order of $100,000$, whereas the 2SLS estimator
was still poorly behaved in terms of bias and Wald test size.

MOP proposed the effective F-statistic
\begin{align}
\hat{F}_{\text{eff}} & =\frac{x'P_{Z}x}{\text{tr}\left(\left(\frac{1}{n}Z'Z\right)^{-1/2}\hat{W}_{2}\left(\frac{1}{n}Z'Z\right)^{-1/2}\right)}\nonumber \\
 & =\frac{x'P_{Z}x}{\text{tr}\left(\hat{W}_{2}\left(\frac{1}{n}Z'Z\right)^{-1}\right)},\label{eq:Feff}
\end{align}
and showed that this F-statistic can be used as a test for weak instruments
in relation to the \citet{Nagar1959} bias of the 2SLS and LIML estimators,
relative to a worst-case benchmark. \citet{AndrewsStockSun2019} advocate
the use of $\hat{F}_{\text{eff}}$ to gauge instrument strength for
the 2SLS estimator. Although this weak-instrument test is related
to the bias, the results presented in \citet[Section 3]{AndrewsStockSun2019}
for a sample of 106 specifications from papers published in the \textit{American
Economic Review} suggest that the effective F-statistic ``may convey
useful information about the instrument strength more broadly, since
we see that conventional asymptotic approximations appear reasonable
in specifications where the effective F-statistic exceeds 10.'',
\citet[p 739]{AndrewsStockSun2019}.

In the next section, we introduce a class of generalized effective
F-statistics, denoted $\hat{F}_{\text{geff}}$, associated with a
class of linear Generalized Method of Moments (GMM) estimators. We
show that the weak-instrument Nagar bias results of MOP, derived for
the effective F-statistic in relation to the 2SLS estimator, applies
to this general class. The robust F-statistic is a member of this
class and the associated GMM estimator, denoted GMMf, is defined as
\begin{equation}
\hat{\beta}_{gmmf}=\frac{x'Z\hat{W}_{2}^{-1}Z'y}{x'Z\hat{W}_{2}^{-1}Z'x},\label{eq:gmmf}
\end{equation}
where the extension ``f'' is for first stage, as the weight matrix
is based on the first-stage residuals. As we will show and explain
below, for the \citet{IAndrewsREStat2018} design with large values
for the robust F-statistic, but small values for the effective F-statistic,
this estimator is much better behaved in terms of bias and also inference
than the 2SLS estimator.

Like the expression of the 2SLS estimator in (\ref{eq:2slsfr}), we
can write the GMMf estimator as
\begin{equation}
\hat{\beta}_{gmmf}=\frac{\hat{\pi}'\left(V\hat{a}r_{r}\left(\hat{\pi}\right)\right)^{-1}\hat{\pi}_{y}}{\hat{\pi}'\left(V\hat{a}r_{r}\left(\hat{\pi}\right)\right)^{-1}\hat{\pi}}=\frac{\hat{\pi}'\left(V\hat{a}r_{r}\left(\hat{\pi}\right)\right)^{-1}\hat{\pi}_{y}}{k_{z}\hat{F}_{r}},\label{eq:gmmffr}
\end{equation}
where $V\hat{a}r_{r}\left(\hat{\pi}\right)=n\left(Z'Z\right)^{-1}\hat{W}_{2}\left(Z'Z\right)^{-1}$
is the robust estimator of the variance of $\hat{\pi}$.

The discussion in \citet[pp 738-739]{AndrewsStockSun2019} provides
the intuition of why $\hat{F}_{\text{eff}}$ is an appropriate statistic
for testing instrument strength when using 2SLS. As they argue, 2SLS
behaves badly when its denominator $x'P_{Z}x$ is close to zero. $\hat{F}_{\text{eff}}$
measures this object, as $x'P_{Z}x$ is its numerator, and gets the
standard errors right on average in the case of general heteroskedasticity,
unlike the non-robust $\hat{F}$. The robust F-statistic $\hat{F}_{r}$
measures a different object, $x'Z\hat{W}_{2}^{-1}Z'x$, and, asymptotically,
``while it has a noncentral chi-square distribution, its noncentrality
parameter does not correspond to the distribution of $\hat{\beta}_{2sls}$'',
\citet[p 739]{AndrewsStockSun2019}. But $x'Z\hat{W}_{2}^{-1}Z'x$
is the denominator of the GMMf estimator, and so $\hat{F}_{r}$ is
the appropriate statistic for testing instrument strength when using
$\hat{\beta}_{gmmf}$.

\section{The Generalized Effective F-Statistic as a Test for Weak Instruments}

\label{sec:Fgeff}

Consider the class of linear Generalized Methods of Moments (GMM)
estimators of $\beta$, given by
\begin{equation}
\hat{\beta}_{\Omega_{n}}=\frac{x'Z\Omega_{n}Z'y}{x'Z\Omega_{n}Z'x},\label{eq:bgmm}
\end{equation}
where $\Omega_{n}$ is a $k_{z}\times k_{z}$ possibly data dependent
weight matrix. It satisfies the following assumption,

\begin{assumption}

\label{Ass:Om}

Under the conditions stated in Assumption \ref{Ass:all}, as $n\rightarrow\infty$,
$\Omega_{n}\stackrel{p}{\rightarrow}\Omega$, with $\Omega$ a finite,
full rank matrix.

\end{assumption}

Assumption \ref{Ass:Om} precludes the standard two-step GMM estimator
with $\Omega_{n}=\Omega_{n}\left(\hat{\beta}_{1}\right)$, e.g.\
in the cross-sectional setting, $\Omega_{n}\left(\hat{\beta}_{1}\right)=\left(\frac{1}{n}\sum_{i=1}^{n}\hat{u}_{1i}^{2}z_{i}z_{i}'\right)^{-1}$,
where $\hat{u}_{1i}=y_{i}-x_{i}\hat{\beta}_{1}$ and $\hat{\beta}_{1}$
is an initial estimator, for example the 2SLS estimator. The initial
estimator is consistent under standard strong-instruments asymptotics,
but it converges to a non-degenerate random variable under weak-instruments
asymptotics, see (\ref{eq:betast}) below, and hence Assumption \ref{Ass:Om}
does not hold.

Let
\[
W_{\Omega}=\left[\begin{array}{cc}
W_{\Omega,1} & W_{\Omega,12}\\
W_{\Omega,12}^{\prime} & W_{\Omega,2}
\end{array}\right]=\left(I_{2}\otimes\Omega^{1/2}\right)W\left(I_{2}\otimes\Omega^{1/2}\right),
\]
and
\begin{equation}
\hat{W}_{\Omega_{n}}=\left(I_{2}\otimes\Omega_{n}^{1/2}\right)\hat{W}\left(I_{2}\otimes\Omega_{n}^{1/2}\right).\label{eq:WhatOmn}
\end{equation}
Then consider the class of generalized effective F-statistics, given
by
\begin{equation}
\hat{F}_{\text{geff}}\left(\Omega_{n}\right)=\frac{x'Z\Omega_{n}Z'x}{n\text{tr}\left(\hat{W}_{\Omega_{n},2}\right)}=\frac{x'Z\Omega_{n}Z'x}{n\text{tr}\left(\Omega_{n}^{1/2}\hat{W}_{2}\Omega_{n}^{1/2}\right)}=\frac{x'Z\Omega_{n}Z'x}{n\text{tr}\left(\hat{W}_{2}\Omega_{n}\right)}.\label{eq:geffF}
\end{equation}
For the 2SLS estimator we have $\Omega_{n}=\left(\frac{1}{n}Z'Z\right)^{-1}$
and
\[
\hat{F}_{\text{geff}}\left(\left(Z'Z/n\right)^{-1}\right)=\frac{x'P_{Z}x}{\text{tr}\left(\hat{W}_{2}\left(Z'Z/n\right)^{-1}\right)}=\hat{F}_{\text{eff}}.
\]
For the GMMf estimator as defined in (\ref{eq:gmmf}), we have $\Omega_{n}=\hat{W}_{2}^{-1}$.
Therefore
\[
\hat{F}_{\text{geff}}\left(\hat{W}_{2}^{-1}\right)=\frac{x'Z\hat{W}_{2}^{-1}Z'x}{nk_{z}}=\hat{F}_{r}.
\]

This remainder of this section together with the proofs in the Appendix
draw heavily on \citet{OleaPflueger2013}. We show that the MOP weak-instruments
testing methodology for 2SLS applies to the class of generalized F-statistics
in relation to the Nagar bias of the linear GMM estimators. We do
this by restating parts of their Lemma 1 and Theorem 1, \citet[p 362]{OleaPflueger2013},
that directly apply to our GMM setting.
\begin{lem}
\label{Lem:Distr} Under Assumptions \ref{Ass:all} and \ref{Ass:Om},
the following limits hold jointly as $n\rightarrow\infty$.
\begin{align}
\hat{\beta}_{\Omega_{n}}-\beta\overset{d}{\rightarrow}\beta_{\Omega}^{*} & =\left(\gamma_{\Omega,2}^{\prime}\gamma_{\Omega,2}\right)^{-1}\gamma_{\Omega,2}^{\prime}\left(\gamma_{\Omega,1}-\beta\gamma_{\Omega,2}\right)\label{eq:betast}\\
\hat{F}_{\text{geff}}\left(\Omega_{n}\right)\overset{d}{\rightarrow}F_{\text{geff}}^{*}\left(\Omega\right) & =\gamma_{\Omega,2}^{\prime}\gamma_{\Omega,2}/\text{tr}\left(W_{\Omega,2}\right)=\gamma_{\Omega,2}^{\prime}\gamma_{\Omega,2}/\text{tr}\left(W_{2}\Omega\right),\label{eq:Feffst}
\end{align}
where
\[
\left(\begin{array}{c}
\gamma_{\Omega,1}\\
\gamma_{\Omega,2}
\end{array}\right)\sim N\left(\left(\begin{array}{c}
c_{\Omega}\beta\\
c_{\Omega}
\end{array}\right),W_{\Omega}\right),
\]
with $c_{\Omega}=\Omega^{1/2}Q_{zz}c$.
\end{lem}
\begin{proof}
See Appendix \ref{subsec:Proof-of-Lemma 1}.
\end{proof}

\subsection{Nagar Bias Approximation}

As $v_{1}=u+\beta v_{2}$, it follows from Assumption \ref{Ass:all}
that, as $n\rightarrow\infty$, 
\[
\frac{1}{\sqrt{n}}\left(\begin{array}{c}
\Omega_{n}^{1/2}Z'u\\
\Omega_{n}^{1/2}Z'v_{2}
\end{array}\right)\overset{d}{\rightarrow}N\left(0,S\left(\beta,W_{\Omega}\right)\right),
\]
with
\begin{equation}
S\left(\beta,W_{\Omega}\right)=\left[\begin{array}{cc}
S_{1}\left(\beta,W_{\Omega}\right) & S_{12}\left(\beta,W_{\Omega}\right)\\
S_{12}\left(\beta,W_{\Omega}\right)' & W_{\Omega,2}
\end{array}\right],\label{eq:SbW}
\end{equation}
where
\begin{align*}
S_{1}\left(\beta,W_{\Omega}\right) & =W_{\Omega,1}-\beta\left(W_{\Omega,12}+W_{\Omega,12}^{\prime}\right)+\beta^{2}W_{\Omega,2}\\
S_{12}\left(\beta,W_{\Omega}\right) & =W_{\Omega,12}-\beta W_{\Omega,2}.
\end{align*}

We can now state Theorem 1 of \citet[p 262]{OleaPflueger2013} on
the Nagar bias approximation for our general class of GMM estimators.
\begin{thm}
\label{Theor:Nagar}Nagar Bias Approximation. Let $c_{\Omega,0}:=c_{\Omega}/\Vert c_{\Omega}\Vert$,
where $\Vert c_{\Omega}\Vert=\sqrt{c_{\Omega}^{\prime}c_{\Omega}}$,
and let the concentration parameter $\mu_{\Omega}^{2}:=\Vert c_{\Omega}\Vert^{2}/\text{tr}\left(W_{\Omega,2}\right)$.
Define the benchmark bias as
\begin{equation}
\text{BM}\left(\beta,W_{\Omega}\right)\coloneqq\sqrt{\frac{\text{tr}\left(S_{1}\left(\beta,W_{\Omega}\right)\right)}{\text{tr}\left(W_{\Omega,2}\right)}}.\label{eq:BM}
\end{equation}
 The Taylor series expansion of $\beta_{\Omega}^{*}$ around $\mu_{\Omega}^{-1}=0$
results in the \citet{Nagar1959} bias approximation
\begin{equation}
\mathbb{E}\left[\beta_{\Omega}^{*}\right]\approx N\left(\beta,c_{\Omega},W_{\Omega}\right)=\frac{n\left(\beta,c_{\Omega,0},W_{\Omega}\right)}{\mu_{\Omega}^{2}},\label{eq:NagBias}
\end{equation}
with
\begin{equation}
n\left(\beta,c_{\Omega,0},W_{\Omega}\right)=\frac{\text{tr}\left(S_{12}\left(\beta,W_{\Omega}\right)\right)-2c_{\Omega,0}^{\prime}S_{12}\left(\beta,W_{\Omega}\right)c_{\Omega,0}}{\text{tr}\left(W_{\Omega,2}\right)}.\label{eq:biasn}
\end{equation}
Further,
\begin{equation}
B\left(W_{\Omega}\right)\coloneqq\sup_{\beta\in\mathbb{R},c_{\Omega,0}\in\mathcal{S}^{k_{z}-1}}\left(\frac{\left|n\left(\beta,c_{\Omega,0},W_{\Omega}\right)\right|}{\text{BM}\left(\beta,W_{\Omega}\right)}\right)\leq1\label{eq:BWst}
\end{equation}
where $\mathcal{S}^{k_{z}-1}$ is the $k_{z}$-dimensional unit sphere.
\end{thm}
\begin{proof}
For the Nagar bias approximation, see Appendix \ref{subsec:Proof-of-Nagar Bias}.
For $B\left(W_{\Omega}\right)\leq1$, see \citet[Apendix A.3]{OleaPflueger2013}.
\end{proof}
MOP interpret the benchmark $\text{BM}\left(\beta,W_{\Omega}\right)$
as a ``worst-case'' bias. It is derived by approximating the expectation
of the ratio by the ratio of expectations, 
\begin{align*}
\mathbb{E}\left[\beta_{\Omega}^{*}\right] & \approx\frac{\mathbb{E}\left[\gamma_{\Omega,2}^{\prime}\left(\gamma_{\Omega,1}-\beta\gamma_{\Omega,2}\right)\right]}{\mathbb{E}\left[\gamma_{\Omega,2}^{\prime}\gamma_{\Omega,2}\right]}=\frac{\text{tr}\left(S_{12}\left(\beta,W_{\Omega}\right)\right)}{\text{tr}\left(W_{\Omega,2}\right)\left(1+\mu_{\Omega}^{2}\right)}\\
 & =\frac{1}{\left(1+\text{\ensuremath{\mu_{\Omega}^{2}}}\right)}\frac{\text{tr}\left(S_{12}\left(\beta,W_{\Omega}\right)\right)}{\sqrt{\text{tr}\left(W_{\Omega,2}\right)}\sqrt{\text{tr}\left(S_{1}\left(\beta,W_{\Omega}\right)\right)}}\sqrt{\frac{\text{tr}\left(S_{1}\left(\beta,W_{\Omega}\right)\right)}{\text{tr}\left(W_{\Omega,2}\right)}.}
\end{align*}
This expected (absolute) bias expression is maximized when the concentration
parameter $\mu_{\Omega}^{2}=0$ and when the the first-stage and structural
errors are perfectly correlated, and it follows that then $\mathbb{E}\left[\beta_{\Omega}^{*}\right]\leq\sqrt{\text{tr}\left(S_{1}\left(\beta,W_{\Omega}\right)\right)/\text{tr}\left(W_{\Omega,2}\right)}$,
see \citet[pp 362-363]{OleaPflueger2013}.

The results of Lemma \ref{Lem:Distr} and Theorem \ref{Theor:Nagar}
are those in MOP for the 2SLS estimator, with $\Omega_{n}=\left(\frac{1}{n}Z'Z\right)^{-1}$
and $\Omega=Q_{zz}^{-1}$. The Lemma and Theorem are replicated here
to show that their methodology applies directly to the larger class
of GMM estimators (\ref{eq:bgmm}), under Assumption \ref{Ass:Om}.
We can therefore also apply the MOP procedure for testing for weak
instruments directly, as we describe next.

\subsection{Null Hypothesis of Weak Instruments and Testing Procedure}

The null hypothesis of weak instruments is specified as in MOP as
\[
H_{0}:\mu_{\Omega}^{2}\in\mathcal{H}\left(W_{\Omega},\tau\right)\,\,\,\text{against}\,\,\,H_{1}:\mu_{\Omega}^{2}\ensuremath{\notin}\mathcal{H}\left(W_{\Omega},\tau\right),
\]
where
\[
\mathcal{H}\left(W_{\Omega},\tau\right)=\left\{ \mu_{\Omega}^{2}\in\mathbb{R}_{+}:\sup_{\beta\in\mathbb{R},c_{\Omega,0}\in\mathcal{S}^{k_{z}-1}}\left(\frac{\left|N\left(\beta,\mu_{\Omega}\sqrt{\text{tr}\left(W_{\Omega,2}\right)}c_{\Omega,0},W_{\Omega}\right)\right|}{\text{BM}\left(\beta,W_{\Omega}\right)}\right)>\tau\right\} ,
\]
or equivalently
\[
\mathcal{H}\left(W_{\Omega},\tau\right)=\left\{ \mu_{\Omega}^{2}\in\mathbb{R}_{+}:\mu_{\Omega}^{2}<\frac{B\left(W_{\Omega}\right)}{\tau}\right\} .
\]
Under the null hypothesis of weak instruments, the Nagar bias exceeds
a fraction $\tau$ of the benchmark for at least some value of the
structural parameter $\beta$ and some direction of the first-stage
parameters, $c_{\Omega,0}$. The parameter $\tau$ is a user specified
threshold, commonly set to $\tau=0.10$.

The generalization of the MOP test for weak instruments is then based
on $\hat{F}_{\text{geff}}\left(\Omega_{n}\right)$ which is asymptotically
distributed as $\gamma_{\Omega,2}^{\prime}\gamma_{\Omega,2}/\text{tr}\left(W_{\Omega,2}\right)$,
with $\gamma_{\Omega,2}\sim N\left(c_{\Omega},W_{\Omega,2}\right)$,
which has mean $1+\mu_{\Omega}^{2}$. It follows that we reject $H_{0}$
when $\hat{F}_{\text{geff}}\left(\Omega_{n}\right)$ is large. Denote
by $F_{c_{\Omega},W_{\Omega,2}}^{-1}\left(\alpha\right)$ the upper
$\alpha$ quantile of the distribution of $\gamma_{\Omega,2}^{\prime}\gamma_{\Omega,2}/\text{tr}\left(W_{\Omega,2}\right)$
and let
\[
cv\left(\alpha,W_{\Omega,2},d_{\Omega}\right):=\sup_{c_{\Omega}\in\mathbb{R}^{k_{z}}}\left\{ F_{c_{\Omega},W_{\Omega,2}}^{-1}\left(\alpha\right)1_{\left(\frac{c_{\Omega}^{\prime}c_{\Omega}}{\text{tr}\left(W_{\Omega,2}\right)}<d_{\Omega}\right)}\right\} ,
\]
where $1_{\left(A\right)}$ denotes the indicator function over a
set $A$. The null of weak instruments is then rejected if 
\[
\hat{F}_{\text{geff}}\left(\Omega_{n}\right)>cv\left(\alpha,\hat{W}_{\Omega_{n},2},B\left(\hat{W}_{\Omega_{n}}\right)/\tau\right),
\]
which is shown in Lemma 2 of \citet[p 363]{OleaPflueger2013} to be
pointwise asymptotically valid,
\[
\sup_{\mathcal{H}\left(W_{\Omega},\tau\right)}\lim_{n\rightarrow\infty}\mathbb{P}\left(\hat{F}_{\text{geff}}\left(\Omega_{n}\right)>cv\left(\alpha,\hat{W}_{\Omega_{n},2},B\left(\hat{W}_{\Omega_{n}}\right)/\tau\right)\right)\leq\alpha,
\]
and, provided that $B\left(\hat{W}_{\Omega_{n}}\right)$ is bounded
in probability,
\[
\lim_{\mu_{\Omega}^{2}\rightarrow\infty}\lim_{n\rightarrow\infty}\mathbb{P}\left(\hat{F}_{\text{geff}}\left(\Omega_{n}\right)>cv\left(\alpha,\hat{W}_{\Omega_{n},2},B\left(\hat{W}_{\Omega_{n}}\right)/\tau\right)\right)=1.
\]

After obtaining $B\left(\hat{W}_{\Omega_{n}}\right)$ by a numerical
routine, MOP show that the critical values can be obtained by Monte
Carlo methods or by the \citet{Patnaik1949} curve-fitting methodology.
The Patnaik critical value is obtained as the the upper $\alpha$
quantile of $\chi_{\hat{k}_{\text{geff}}\left(\Omega_{n}\right)}^{2}\left(d_{\Omega_{n},\tau}\hat{k}_{\text{geff}}\left(\Omega_{n}\right)\right)/\hat{k}_{\text{geff}}\left(\Omega_{n}\right)$
where $\chi_{\hat{k}_{\text{geff}}\left(\Omega_{n}\right)}^{2}\left(d_{\Omega_{n},\tau}\hat{k}_{\text{geff}}\left(\Omega_{n}\right)\right)$
denotes the noncentral $\chi^{2}$ distribution with $\hat{k}_{\text{geff}}$$\left(\Omega_{n}\right)$
degrees of freedom and noncentrality parameter $d_{\Omega_{n},\tau}\hat{k}_{\text{geff}}\left(\Omega_{n}\right)$,
with
\begin{align}
d_{\Omega_{n}\tau} & =B\left(\hat{W}_{\Omega_{n}}\right)/\tau;\label{eq:dOmn}\\
\hat{k}_{\text{geff}}\left(\Omega_{n}\right) & =\frac{\left[\text{tr}\left(\hat{W}_{\Omega_{n},2}\right)\right]^{2}\left(1+2d_{\Omega_{n},\tau}\right)}{\text{tr}\left(\hat{W}_{\Omega_{n},2}^{\prime}\hat{W}_{\Omega_{n},2}\right)+2d_{\Omega_{n},\tau}\text{tr}\left(\hat{W}_{\Omega_{n},2}\right)\lambda_{\text{max}}\left(\hat{W}_{\Omega_{n},2}\right)},\label{eq:keff}
\end{align}
and where $\lambda_{\text{max}}\left(\hat{W}_{\Omega_{n},2}\right)$
denotes the maximum eigenvalue of $\hat{W}_{\Omega_{n},2}$.

To summarize, and following MOP, the weak-instruments test procedure
related to the Nagar approximation of the bias of the GMM estimator
$\hat{\beta}_{\Omega_{n}}$as defined in (\ref{eq:bgmm}), under Assumptions
\ref{Ass:all} and \ref{Ass:Om} is as follows.
\begin{enumerate}
\item Compute the generalized effective F-statistic, 
\[
\hat{F}_{\text{geff}}\left(\Omega_{n}\right)=\frac{x'Z\Omega_{n}Z'x}{n\text{tr}\left(\hat{W}_{2}\Omega_{n}\right)}.
\]
\item \label{step:BW} Obtain
\[
B\left(\hat{W}_{\Omega_{n}}\right)=\sup_{\beta\in\mathbb{R},c_{\Omega,0}\in\mathcal{S}^{k_{z}-1}}\left(\frac{\left|n\left(\beta,c_{\Omega,0},\hat{W}_{\Omega_{n}}\right)\right|}{\text{BM}\left(\beta,\hat{W}_{\Omega_{n}}\right)}\right)
\]
by a numerical maximization routine, where $\hat{W}_{\Omega_{n}}$
is as defined in (\ref{eq:WhatOmn}), $n\left(\beta,c_{\Omega,0},\hat{W}_{\Omega_{n}}\right)$
and $\text{BM}\left(\beta,\hat{W}_{\Omega_{n}}\right)$ in (\ref{eq:biasn})
and (\ref{eq:BWst}) respectively, with the estimated $\hat{W}_{\Omega_{n}}$
replacing $W_{\Omega}$.
\item \label{step:Patnaik} \begin{sloppypar} Applying the \citet{Patnaik1949}
curve-fitting methodology, estimate the effective degrees of freedom
$\hat{k}_{\text{geff}}\left(\Omega_{n}\right)$ as given in (\ref{eq:keff})
and compute the critical value $cv\left(\alpha,\hat{W}_{\Omega_{n},2},d_{\Omega_{n},\tau}\right)$
for a user specified threshold value $\tau$ as the upper $\alpha$
quantile of $\chi_{\hat{k}_{\text{geff}}\left(\Omega_{n}\right)}^{2}\left(d_{\Omega_{n},\tau}\hat{k}_{\text{geff}}\left(\Omega_{n}\right)\right)/\hat{k}_{\text{geff}}\left(\Omega_{n}\right)$,
where $d_{\Omega_{n},\tau}$ is defined in (\ref{eq:dOmn}).\end{sloppypar} 
\item \label{step:rej} Reject the null of weak instruments, that the proportion
of the Nagar approximation of the bias of $\hat{\beta}_{\Omega_{n}}$
relative to the benchmark bias is larger than $\tau$, if $\hat{F}_{\text{geff}}\left(\Omega_{n}\right)>cv\left(\alpha,\hat{W}_{\Omega_{n},2},d_{\Omega_{n},\tau}\right)$. 
\end{enumerate}
As an alternative to step \ref{step:Patnaik}.\ one can use Monte
Carlo methods, see \citet[Section 5]{OleaPflueger2013}. As $B\left(W_{\Omega}\right)\leq1$,
MOP propose a simplified asymptotically valid but conservative test.
For this simplified test procedure, step \ref{step:BW}.\ is not
needed, instead replacing $B\left(\hat{W}_{\Omega_{n}}\right)$ by
$1$, and so $d_{\Omega_{n},\tau}$ by $1/\tau$, in steps \ref{step:Patnaik}.\
and \ref{step:rej}. Hence the simplified test rejects the null hypothesis
of weak instruments if
\[
\hat{F}_{\text{geff}}>cv\left(\alpha,\hat{W}_{\Omega_{n},2},1/\tau\right),
\]
with $cv\left(\alpha,\hat{W}_{\Omega_{n},2},1/\tau\right)\geq cv\left(\alpha,\hat{W}_{\Omega,2},d_{\Omega_{n},\tau}\right)$.

\section{The Robust F-Statistic as a Test for Weak Instruments }

\label{sec:RobF}

For the robust F-statistic in relation to the Nagar bias of the GMMf
estimator as defined in (\ref{eq:gmmf}), 
\[
\hat{\beta}_{gmmf}=\frac{x'Z\hat{W}_{2}^{-1}Z'y}{x'Z\hat{W}_{2}^{-1}Z'x},
\]
the above expressions apply, but simplify significantly. With $\Omega_{n}=\hat{W}_{2}^{-1}\stackrel{p}{\rightarrow}W_{2}^{-1}=\Omega$,
it follows that
\[
W_{\Omega}=\left[\begin{array}{cc}
W_{2}^{-1/2}W_{1}W_{2}^{-1/2} & W_{2}^{-1/2}W_{12}W_{2}^{-1/2}\\
W_{2}^{-1/2}W_{12}'W_{2}^{-1/2} & I_{k_{z}}
\end{array}\right].
\]
For the Nagar bias expression, we get
\[
N_{gmmf}\left(\beta,c_{\Omega},W_{\Omega}\right)=\frac{1}{\mu_{\Omega}^{2}}n_{gmmf}\left(\beta,c_{\Omega},W_{\Omega}\right),
\]
with $\mu_{\Omega}^{2}=c'Q_{zz}W_{2}^{-1}Q_{zz}c/k_{z}$ and
\[
n_{gmmf}\left(\beta,c_{\Omega,0},W_{\Omega}\right)=\frac{\text{tr}\left(S_{12}\left(\beta,W_{\Omega}\right)\right)-2c_{\Omega,0}^{\prime}S_{12}\left(\beta,W_{\Omega}\right)c_{\Omega,0}}{k_{z}},
\]
where
\[
S_{12}\left(\beta,W_{\Omega}\right)=W_{\Omega,12}-\beta I_{k_{z}}.
\]
The expression therefore simplifies to
\[
n_{gmmf}\left(\beta,c_{\Omega,0},W_{\Omega}\right)=\frac{1}{k_{z}}\left(\text{tr}\left(W_{\Omega,12}\right)-2c_{\Omega,0}^{\prime}W_{\Omega,12}c_{\Omega,0}-\left(k_{z}-2\right)\beta\right).
\]
The benchmark worst-case bias for the GMMf estimator is then given
by
\[
\text{BM}\left(\beta,W_{\Omega}\right)=\sqrt{\frac{\text{tr}\left(S_{1}\left(\beta,W_{\Omega}\right)\right)}{k_{z}}}
\]
where
\[
\text{tr}\left(S_{1}\left(\beta,W_{\Omega}\right)\right)=\text{tr}\left(W_{\Omega,1}\right)-2\beta\text{tr}\left(W_{\Omega,12}\right)+k_{z}\beta^{2}.
\]
Then
\begin{align*}
B_{gmmf}\left(W_{\Omega}\right) & =\sup_{\beta\in\mathbb{R},c_{\Omega,0}\in\mathcal{S}^{k_{z}-1}}\left(\frac{\left|n_{gmmf}\left(\beta,c_{\Omega,0},W_{\Omega}\right)\right|}{\text{BM}\left(\beta,W_{\Omega}\right)}\right)\\
 & =\sup_{\beta\in\mathbb{R},c_{\Omega,0}\in\mathcal{S}^{k_{z}-1}}\left(\frac{\left|\text{tr}\left(W_{\Omega,12}\right)-2c_{\Omega,0}^{\prime}W_{\Omega,12}c_{\Omega,0}-\left(k_{z}-2\right)\beta\right|}{\sqrt{k_{z}\left(\text{tr}\left(W_{\Omega,1}\right)-2\beta\text{tr}\left(W_{\Omega,12}\right)+k_{z}\beta^{2}\right)}}\right)
\end{align*}
As $\Omega=W_{2}^{-1}$ we have that $\gamma_{\Omega,2}\sim N\left(c_{\Omega},I_{k_{z}}\right)$
and so it follows that $\gamma_{\Omega,2}^{\prime}\gamma_{\Omega,2}\sim\chi_{k_{z}}^{2}\left(c_{\Omega}^{\prime}c_{\Omega}\right)=\chi_{k_{z}}^{2}\left(k_{z}\mu_{\Omega}^{2}\right)$.
Therefore, the null of weak instruments for the GMMf estimator, specified
as
\[
H_{0}:\mu_{\Omega}^{2}\in\mathcal{H}_{gmmf}\left(W_{\Omega},\tau\right),
\]
\[
\mathcal{H}_{gmmf}\left(W_{\Omega},\tau\right)=\left\{ \mu_{\Omega}^{2}\in\mathbb{R}_{+}:\mu_{\Omega}^{2}<\frac{B_{gmmf}\left(W_{\Omega}\right)}{\tau}\right\} ,
\]
is rejected if
\[
\hat{F}_{r}>cv\left(\alpha,k_{z},d_{\Omega_{n},\tau}\right),
\]
with $cv\left(\alpha,k_{z},d_{\Omega_{n},\tau}\right)$ the upper
$\alpha$ quantile of $\chi_{k_{z}}^{2}\left(k_{z}d_{\Omega_{n}}\right)/k_{z}$,
and where $d_{\Omega_{n},\tau}=B_{gmmf}\left(\hat{W}_{\Omega_{n}}\right)/\tau$
and $\hat{W}_{\Omega_{n}}=\left(I_{2}\otimes\hat{W}_{2}^{-1/2}\right)\hat{W}\left(I_{2}\otimes\hat{W}_{2}^{-1/2}\right)$.
Relative to the general results for the generalized effective F-statistic,
we see that for the GMMf estimator, $B_{gmmf}\left(\hat{W}_{\Omega_{n}}\right)$
is a simpler function to maximize with respect to $\beta$ and $c_{\Omega,0}$
than $B\left(\hat{W}_{\Omega_{n}}\right)$ from (\ref{eq:BWst}) when
$\Omega_{n}\neq\hat{W}_{2}^{-1}$. There is further no need for Monte
Carlo simulations or Patnaik's curve-fitting methodology to compute
the critical values, as $\hat{F}_{r}$ follows an asymptotic scaled
noncentral chi-square distribution, with the ``effective'' degrees
of freedom here equal to $k_{z}$.

To summarize the weak-instruments testing procedure in relation to
the approximate Nagar bias of $\hat{\beta}_{gmmf}$,
\begin{enumerate}
\item Compute the robust F-statistic
\[
\hat{F}_{r}=\frac{x'Z\hat{W}_{2}^{-1}Z'x}{nk_{z}}.
\]
\item Obtain
\[
B_{gmmf}\left(\hat{W}_{\Omega_{n}}\right)=\sup_{\beta\in\mathbb{R},c_{\Omega,0}\in\mathcal{S}^{k_{z}-1}}\left(\frac{\left|\text{tr}\left(\hat{W}_{\Omega_{n},12}\right)-2c_{\Omega,0}^{\prime}\hat{W}_{\Omega_{n},12}c_{\Omega,0}-\left(k_{z}-2\right)\beta\right|}{\sqrt{k_{z}\left(\text{tr}\left(\hat{W}_{\Omega_{n},1}\right)-2\beta\text{tr}\left(\hat{W}_{\Omega_{n},12}\right)+k_{z}\beta^{2}\right)}}\right)
\]
by a numerical maximization routine, where $\hat{W}_{\Omega_{n}}=\left(I_{2}\otimes\hat{W}_{2}^{-1/2}\right)\hat{W}\left(I_{2}\otimes\hat{W}_{2}^{-1/2}\right)$.
\item \begin{sloppypar} Reject the null of weak instruments if $\hat{F}_{r}>cv\left(\alpha,\hat{W}_{\Omega_{n},2},d_{\Omega_{n},\tau}\right)$,
where $d_{\Omega_{n},\tau}=B_{gmmf}\left(\hat{W}_{\Omega_{n}}\right)/\tau$
and where $cv\left(\alpha,\hat{W}_{\Omega_{n},2},d_{\Omega_{n},\tau}\right)$
is the upper $\alpha$ quantile of $\chi_{k_{z}}^{2}\left(k_{z}d_{\Omega_{n},\tau}\right)/k_{z}$.
\end{sloppypar}
\end{enumerate}
The simplified conservative test based on the fact that $B_{gmmf}\left(W_{\Omega}\right)\leq1$
is obtained using the critical value $cv\left(\alpha,k_{z},1/\tau\right)$,
which is simply the upper $\alpha$ quantile of $\chi_{k_{z}}^{2}\left(k_{z}/\tau\right)/k_{z}$.

We can illustrate divergence of $\hat{F}_{\text{eff}}$ and $\hat{F}_{r}$,
building on examples in \citet{OleaPflueger2013} and \citet{AndrewsStockSun2019}.
Let $k_{z}=2$, $Q_{zz}=I_{2}$ and $W_{2}=\left(\begin{array}{cc}
\xi{}_{1}^{2} & 0\\
0 & \xi_{2}^{2}
\end{array}\right)$. Then for 2SLS/$\hat{F}_{\text{eff}}$ we have the concentration
parameter $\mu_{Q_{zz}^{-1}}^{2}=\frac{c_{1}^{2}+c_{2}^{2}}{\xi_{1}^{2}+\xi_{2}^{2}}$,
and for GMMf/$\hat{F}_{r}$ we have $\mu_{W_{2}^{-1}}^{2}=\frac{1}{2}\left(\frac{c_{1}^{2}}{\xi_{1}^{2}}+\frac{c_{2}^{2}}{\xi_{2}^{2}}\right)$.
Then, for $c_{1}^{2}>0$ and $c_{2}^{2}>0$, if for example $\xi_{1}^{2}\rightarrow0$
and $\xi_{2}^{2}>c_{1}^{2}+c_{2}^{2}$ then $\mu_{Q_{zz}^{-1}}^{2}\rightarrow\frac{c_{1}^{2}+c_{2}^{2}}{\xi_{2}^{2}}<1$
and $\mu_{W_{2}^{-1}}^{2}\rightarrow\infty$. In this situation, there
is no weak-instruments problem for the GMMf estimator, but there is
for the 2SLS estimator. A design where such heteroskedasticity can
be generated is the grouped-data IV one, and is the one used in \citet{IAndrewsREStat2018}.
We consider this design in Section \ref{sec:Grouped-Data-IV-Model},
where we further show how the GMMf estimator utilizes the information
in such designs better than the 2SLS estimator.\footnote{\label{fn:concpar} Note that for this example it is the case that
if $\left(\xi_{1}^{2}>\xi_{2}^{2}\right)$ \textit{and} $\left(\frac{c_{1}^{2}}{\xi_{1}^{2}}>\frac{c_{2}^{2}}{\xi_{2}^{2}}\right)$
then $\mu_{Q_{zz}^{-1}}^{2}>\mu_{W_{2}^{-1}}^{2}$.}

\subsection{Efficiency}

\label{subsec:Efficiency}

Under standard strong-instruments asymptotics, the 2SLS estimator
is consistent and asymptotically efficient if $\mathbb{E}\left[u_{i}^{2}z_{i}z_{i}'\right]=\sigma_{u}^{2}Q_{zz}$,
and for the GMMf estimator this is the case if $\mathbb{E}\left[v_{2,i}^{2}z_{i}z_{i}'\right]=\text{\ensuremath{\delta}}\mathbb{E}\left[u_{i}^{2}z_{i}z_{i}'\right]$,
for some constant $\delta>0$. However, under the weak-instruments
asymptotics setting of Assumption \ref{Ass:all}, it follows from
the expression of the random variable $\beta_{\Omega}^{*}$ in (\ref{eq:betast})
that the GMM estimators are biased and inconsistent. Therefore, the
weak-instruments test is based on the relative Nagar bias, and GMM
estimators can then be ranked on their bias performance, not on their
relative efficiency. As the MOP benchmark bias is estimator specific,
we first propose to harmonize the bias by using the worst-case OLS
bias as a benchmark.

\section{Harmonizing the Benchmark Bias}

\label{sec:Harmon}

The benchmark bias $\text{BM}\left(\beta,W_{\Omega}\right)$ as defined
in (\ref{eq:BM}) is estimator/F-statistic specific and is derived
using an ad hoc approximation of $\mathbb{E}\left(\beta_{\Omega}^{*}\right)$.
This makes a comparison of weak-instrument test results between generalized
F-statistics for different GMM estimators difficult, as the critical
values of the test are based on the maximum of the Nagar bias approximation
relative to the estimator-specific benchmark. A solution is to harmonize
the benchmark bias, such that it is the same for each estimator/F-statistic.
For the homoskedastic case, the bias related weak-instruments critical
values of \citet{StockYogo2005} for the nonrobust F-statistic are
in terms of the relative bias, that of the 2SLS estimator relative
to that of the OLS estimator. We propose here to harmonize the benchmark
worst-case bias and to replace it with the worst-case bias of the
OLS estimator.

The probability limit of the worst-case weak-instrument OLS (absolute)
bias, under Assumption \ref{Ass:all} is given by
\begin{align*}
\hat{\beta}_{LS}-\beta & =\frac{x'u}{x'x}\\
 & =\frac{c'Z'u/\sqrt{n}+v_{2}'u}{c'Z'Zc/n+2c'Z'v_{2}/\sqrt{n}+v_{2}'v_{2}}\\
 & \stackrel{p}{\rightarrow}\frac{\sigma_{uv_{2}}\left(\beta,\Sigma_{v}\right)}{\sigma_{2}^{2}}=\rho_{uv_{2}}\frac{\sigma_{u}\left(\beta,\Sigma_{v}\right)}{\sigma_{2}}\\
 & \leq\frac{\sigma_{u}\left(\beta,\Sigma_{v}\right)}{\sigma_{2}}=\sqrt{\frac{\sigma_{1}^{2}-2\beta\sigma_{12}+\beta^{2}\sigma_{2}^{2}}{\sigma_{2}^{2}}}=\text{BM}_{LS}\left(\beta,\Sigma_{v}\right).
\end{align*}

Note that this benchmark bias is only a function of $\beta$ and $\Sigma_{v}$,
the marginal variance of $(v_{1}\,v_{2})'$, applies to all forms
of general heteroskedasticity and does not depend on homoskedasticity
assumptions. However, it can be seen to be a valid benchmark for all
GMM estimators defined in (\ref{eq:bgmm}) by considering the worst-case
benchmark bias $\text{BM}\left(\beta,W_{\Omega}\right)$ under the
maintained assumption of homoskedasticity. In that case we have $W=\Sigma_{v}\otimes Q_{zz}$
and thus $W_{\Omega}=\Sigma_{v}\otimes\Omega^{-1/2}Q_{zz}\Omega^{-1/2}$.
Then it follows that
\begin{align*}
\text{BM}\left(\beta,W_{\Omega}\right) & =\sqrt{\frac{\left(\sigma_{1}^{2}-2\beta\sigma_{12}+\beta^{2}\sigma_{2}^{2}\right)\text{tr}\left(Q_{zz}\Omega^{-1}\right)}{\sigma_{2}^{2}\text{tr}\left(Q_{zz}\Omega^{-1}\right)}}\\
 & =\sqrt{\frac{\sigma_{1}^{2}-2\beta\sigma_{12}+\beta^{2}\sigma_{2}^{2}}{\sigma_{2}^{2}}}=\text{BM}_{LS}\left(\beta,\Sigma_{v}\right).
\end{align*}

We have that
\begin{align*}
\lim_{\beta\rightarrow\pm\infty}\frac{\left|n\left(\beta,c_{\Omega,0},W_{\Omega}\right)\right|}{\text{BM}\left(\beta,W_{\Omega}\right)} & =\lim_{\beta\rightarrow\pm\infty}\frac{\left|n\left(\beta,c_{\Omega,0},W_{\Omega}\right)\right|}{\text{BM}_{LS}\left(\beta,\Sigma_{v}\right)}=1-\frac{2\lambda_{\text{min}}\left(W_{\Omega,2}\right)}{\text{tr}\left(W_{\Omega,2}\right)},
\end{align*}
where $\lambda_{\text{min}}\left(W_{\Omega,2}\right)$ is the minimum
eigenvalue of $W_{\Omega,2}$. It follows further from Assumptions
\ref{Ass:all} and \ref{Ass:Om} that 
\[
0<\frac{\text{BM}\left(\beta,W_{\Omega}\right)}{\text{BM}_{LS}\left(\beta,\Sigma_{v}\right)}<C
\]
$\forall\beta\in\mathbb{R}$, for some finite $C>1$. Note that the
positive definiteness, or full rank assumption of $\Sigma_{v}$ is
important, as otherwise $\text{BM}_{LS}\left(\beta,\Sigma_{v}\right)$
would be zero for some value of $\beta$. Therefore the case of $\rho_{12}^{2}=1$
is excluded. It then follows from Theorem \ref{Theor:Nagar} that,
\[
B_{LS}\left(W_{\Omega},\Sigma_{v}\right)\coloneqq\sup_{\beta\in\mathbb{R},c_{\Omega,0}\in\mathcal{S}^{k_{z}-1}}\left(\frac{\left|n\left(\beta,c_{\Omega,0},W_{\Omega}\right)\right|}{\text{BM}_{LS}\left(\beta,\Sigma_{v}\right)}\right)<\infty
\]
and the test procedures described in Section \ref{sec:Fgeff} apply,
replacing $B\left(W_{\Omega}\right)$ by $B_{LS}\left(W_{\Omega},\Sigma_{v}\right)$
from step \ref{step:BW} onwards in the summaries of the testing procedures.
The null of weak instruments is then rejected if
\[
\hat{F}_{\text{geff}}\left(\Omega_{n}\right)>cv^{LS}\left(\alpha,\hat{W}_{\Omega_{n},2},d_{\Omega_{n},\tau}^{LS}\right),
\]
where $d_{\Omega_{n},\tau}^{LS}=B_{LS}\left(\hat{W}_{\Omega_{n}},\hat{\Sigma}_{v}\right)/\tau$. 

The interpretation of the weak-instruments null hypothesis $H_{0}:\mu_{\Omega}^{2}\in\mathcal{H}\left(W_{\Omega},\Sigma_{v},\tau\right)$,
where
\[
\mathcal{H}\left(W_{\Omega},\Sigma_{v},\tau\right)=\left\{ \mu_{\Omega}^{2}\in\mathbb{R}_{+}:\sup_{\beta\in\mathbb{R},c_{\Omega,0}\in\mathcal{S}^{k_{z}-1}}\left(\frac{\left|N\left(\beta,\mu_{\Omega}\sqrt{\text{tr}\left(W_{\Omega,2}\right)}c_{\Omega,0},W_{\Omega}\right)\right|}{\text{BM}_{LS}\left(\beta,\Sigma_{v}\right)}\right)>\tau\right\} ,
\]
is then that the Nagar bias exceeds a fraction $\tau$ of the benchmark
maximum OLS bias for at least some value of the structural parameter
$\beta$ and some direction of the first-stage coefficients. Whilst
this is not the same as the asymptotic relative bias results of \citet{StockYogo2005},
it is more aligned with it. It makes the interpretation of the null
hypothesis and hence that of its rejection the same for different
GMM estimators and their associated generalized effective F-statistics
and is not based on an ad hoc approximation. In the remainder, this
benchmark is used. Note that the simplified conservative test procedure
does not apply here, as it is not the case that $B_{LS}\left(W_{\Omega},\Sigma_{v}\right)\leq1$
in general.

\section{Grouped-Data IV Model}

\label{sec:Grouped-Data-IV-Model}

We now consider the heteroskedastic model designs from \citet{IAndrewsREStat2018}.
In these designs, very large values of the robust F-statistic are
accompanied by a poor performance of the 2SLS estimator, where \citet{IAndrewsREStat2018}
focused on coverage distortions of confidence sets. We find that in
these designs the effective F-statistic is indeed small, indicating
a weak-instrument problem for 2SLS, but large values of the robust
F-statistic indicate there is not a weak-instrument problem for the
GMMf estimator.

Following the example of divergence between $\hat{F}_{\text{eff}}$
and $\hat{F}_{r}$ given at at the end of Section \ref{sec:RobF},
the design in \citet[Supplementary Appendix C.3]{IAndrewsREStat2018}
is that of a grouped-data IV setup,
\begin{eqnarray*}
y_{i} & = & x_{i}\beta+u_{i}\\
x_{i} & = & z_{i}^{\prime}\pi+v_{2,i},
\end{eqnarray*}
for $i=1,...,n$, where the $G$-vector $z_{i}\in\left\{ e_{1},...,e_{G}\right\} $,
with $e_{g}$ a $G$-vector with $g$th entry equal to $1$ and zeros
everywhere else, for $g=1,\ldots,G$.

The variance-covariance structure for the errors is modeled fully
flexibly by group, and specified as
\[
\left(\left(\begin{array}{c}
u_{i}\\
v_{2,i}
\end{array}\right)|z_{i}=e_{g}\right)\sim\left(0,\Sigma_{g}\right),
\]
\begin{equation}
\Sigma_{g}=\left[\begin{array}{cc}
\sigma_{u,g}^{2} & \sigma_{uv_{2},g}\\
\sigma_{uv_{2},g} & \sigma_{v_{2},g}^{2}
\end{array}\right].\label{eq:Sigmas}
\end{equation}
At the group level, we therefore have for group member $j$ in group
$g$
\begin{eqnarray}
y_{jg} & = & x_{jg}\beta+u_{jg}\label{modyjs}\\
x_{jg} & = & \pi_{g}+v_{2,jg}\label{modxjs}
\end{eqnarray}
\[
\left(\begin{array}{c}
u_{jg}\\
v_{2,jg}
\end{array}\right)\sim\left(0,\Sigma_{g}\right),
\]
for $j=1,...,n_{g}$ and $g=1,...,G$, with $n_{g}$ the number of
observations in group $g$, $\sum_{g=1}^{G}n_{g}=n$, see also \citet{BekkervdPloeg2005}.
We assume that $\lim_{n\rightarrow\infty}\frac{n_{g}}{n}=f_{g}$,
with $0<f_{g}<1$.

The OLS estimator of $\pi_{g}$ is given by $\widehat{\pi}_{g}=\overline{x}_{g}=\frac{1}{n_{g}}\sum_{j=1}^{n_{s}}x_{jg}$
and $Var\left(\widehat{\pi}_{g}\right)=\sigma_{v_{2},g}^{2}/n_{g}$.
The OLS residual is $\widehat{v}_{2,jg}=x_{jg}-\overline{x}_{g}$
and the estimator for the variance is given by $V\widehat{a}r\left(\widehat{\pi}_{g}\right)=\widehat{\sigma}_{v_{2},g}^{2}/n_{g}$,
where $\widehat{\sigma}_{v,g}^{2}=\frac{1}{n_{g}}\sum_{j=1}^{n_{g}}\widehat{v}_{2,jg}^{2}$.
Let $Z$ be the $n\times G$ matrix of instruments. For the vector
$\pi$ the OLS estimator is given by
\[
\widehat{\pi}=\left(Z^{\prime}Z\right)^{-1}Z^{\prime}x=\left(\overline{x}_{1},\overline{x}_{2},...,\overline{x}_{G}\right)^{\prime}.
\]
Let
\begin{eqnarray}
\hat{W}_{2} & = & \frac{1}{n}\sum_{i=1}^{n}\widehat{v}_{2,i}^{2}z_{i}z_{i}^{\prime}\label{omev}\\
 & = & \text{Diag}\left(n_{g}\widehat{\sigma}_{v_{2},g}^{2}\right)/n,\nonumber 
\end{eqnarray}
where $\text{Diag}\left(q_{g}\right)$ is a diagonal matrix with $g$th
diagonal element $q_{g}$. Then the robust estimator of $Var\left(\widehat{\pi}\right)$
is given by
\begin{eqnarray*}
V\widehat{a}r_{r}\left(\widehat{\pi}\right) & = & \left(Z^{\prime}Z\right)^{-1}n\widehat{W}_{2}\left(Z^{\prime}Z\right)^{-1}\\
 & = & \text{Diag}\left(\widehat{\sigma}_{v_{2},g}^{2}/n_{g}\right).
\end{eqnarray*}
The non-robust variance estimator is
\begin{eqnarray*}
V\widehat{a}r\left(\widehat{\pi}\right) & = & \left(\frac{1}{n}\sum_{i=1}^{n}\widehat{v}_{2,i}^{2}\right)\left(Z^{\prime}Z\right)^{-1}\\
 & = & \left(\sum_{g=1}^{G}\frac{n_{g}}{n}\widehat{\sigma}_{v_{2},g}^{2}\right)\text{Diag}\left(\frac{1}{n_{g}}\right).
\end{eqnarray*}

The group- (or instrument-) specific IV estimators for $\beta$ are
given by
\begin{equation}
\widehat{\beta}_{g}=\frac{z_{g}^{\prime}y}{z_{g}^{\prime}x}=\frac{\overline{y}_{g}}{\overline{x}_{g}},\label{betasind}
\end{equation}
with $\overline{y}_{g}=\frac{1}{n_{g}}\sum_{j=1}^{n_{g}}y_{jg}$,
and the 2SLS estimator for $\beta$ is
\begin{eqnarray*}
\widehat{\beta}_{2sls} & = & \left(x^{\prime}P_{Z}x\right)^{-1}x^{\prime}P_{Z}y\\
 & = & \frac{\sum_{g=1}^{G}n_{g}\overline{x}_{g}\overline{y}_{g}}{\sum_{g=1}^{G}n_{g}\overline{x}_{g}^{2}}\\
 & = & \frac{\sum_{g=1}^{G}n_{g}\overline{x}_{g}^{2}\left(\overline{y}_{g}/\overline{x}_{g}\right)}{\sum_{g=1}^{G}n_{g}\overline{x}_{g}^{2}}=\sum_{g=1}^{G}\hat{w}_{2sls,g}\widehat{\beta}_{g}\text{,}
\end{eqnarray*}
the standard result that $\widehat{\beta}_{2sls}$ is a linear combination
of the instrument specific IV estimators, (see e.g.\ \citealp{Windmeijer2019}).
The weights are given by
\begin{equation}
\hat{w}_{2sls,g}=\frac{n_{g}\overline{x}_{g}^{2}}{\sum_{s=1}^{G}n_{s}\overline{x}_{s}^{2}}\geq0\label{ws}
\end{equation}
and hence the 2SLS estimator is here a weighted average of the group
specific estimators.

For the group specific estimates, the first-stage F-statistics are
given by
\begin{equation}
\hat{F}_{g}=\frac{\widehat{\pi}_{g}^{2}}{V\widehat{a}r\left(\widehat{\pi}_{g}\right)}=\frac{n_{g}\overline{x}_{g}^{2}}{\widehat{\sigma}_{v_{2},g}^{2}}\label{Fps}
\end{equation}
for $g=1,...,G$. As the errors $\left(u_{jg},v_{2,jg}\right)$ are
iid within groups, the standard weak-instruments results of \citet{StaigerStock1997}
and \citet{StockYogo2005} apply to each group-specific IV estimator
$\widehat{\beta}_{g}$. As these are just-identified models, we can
relate the values of the F-statistics to Wald-test size distortions.

From (\ref{ws}) and (\ref{Fps}) it follows that the weights for
the 2SLS estimator are related to the individual F-statistics as follows
\begin{equation}
\hat{w}_{2sls,g}=\frac{n_{g}\overline{x}_{g}^{2}}{\sum_{s=1}^{G}n_{s}\overline{x}_{s}^{2}}=\frac{\widehat{\sigma}_{v_{2},g}^{2}\hat{F}_{g}}{\sum_{s=1}^{G}\widehat{\sigma}_{v_{2},s}^{2}\hat{F}_{s}}.\label{eq:w2sls}
\end{equation}
Under first-stage homoskedasticity, $\sigma_{v_{2},g}^{2}=\sigma_{v_{2},s}^{2}$,
for $g,s=1,\ldots,G$, then $\widehat{\sigma}_{v_{2},s}^{2}\approx\widehat{\sigma}_{v_{2},s}^{2}$
for all $g,s$, and hence $\hat{F}\approx\frac{1}{G}\sum_{g=1}^{G}\hat{F}_{g}$.
Then the weights are given by $\hat{w}_{2sls,g}\approx\frac{\hat{F}_{g}}{\sum_{s=1}^{G}\hat{F}_{\pi_{g}}}\approx\frac{\hat{F}_{g}}{G\hat{F}}$,
so we see that the groups with the larger individual F-statistics
get the larger weights in the 2SLS estimator under homoskedasticity.

This is not necessarily the case under heteroskedasticity. For equal
sized groups with approximately the same value of the signal $\hat{\pi}_{g}^{2}$,
$\hat{F}_{g}$ is larger for groups with, and because of, the smaller
values of $\widehat{\sigma}_{v_{2},g}^{2}$. The 2SLS weights ignore
this information and give in this example approximately equal weights
to groups with similar values of $\hat{\pi}_{g}^{2}$ . In practice
it could then be the case that a small variance, large $\hat{F}_{g}$
group could receive a small weight in the 2SLS estimator. As shown
in the Monte Carlo exercises below, this is exactly what happens in
the design of \citet{IAndrewsREStat2018}. There is one group with
a large individual F-statistic. However, this group has a very small
population variance $\sigma_{v_{2},g}^{2}$ resulting in a relatively
small weight in the 2SLS estimator and a poor performance of the estimator
in terms of  bias and size of the Wald-test.

The non-robust F-statistic for $\pi$ is given by
\begin{eqnarray*}
\hat{F} & = & \frac{1}{G}\widehat{\pi}^{\prime}\left(V\widehat{a}r\left(\widehat{\pi}\right)\right)^{-1}\widehat{\pi}\\
 & = & \frac{1}{G}\frac{\sum_{g=1}^{G}n_{g}\overline{x}_{g}^{2}}{\left(\sum_{g=1}^{G}\frac{n_{g}}{n}\widehat{\sigma}_{v_{2},g}^{2}\right)}=\frac{1}{G}\sum_{g=1}^{G}\frac{\widehat{\sigma}_{v_{2},g}^{2}}{\left(\sum_{s=1}^{G}\frac{n_{s}}{n}\widehat{\sigma}_{v_{2},s}^{2}\right)}\hat{F}_{g}.
\end{eqnarray*}
The effective F-statistic is given by
\[
\hat{F}_{\text{eff}}=\sum_{g=1}^{G}\frac{\widehat{\sigma}_{v_{2},g}^{2}}{\left(\sum_{s=1}^{G}\widehat{\sigma}_{v_{2},s}^{2}\right)}\hat{F}_{g},
\]
and so $\hat{F}_{\text{eff}}=\hat{F}$ if groups sizes are equal,
$n_{g}=n/G$ for $g=1,\ldots,G$. In the designs of \citet{IAndrewsREStat2018}
group sizes are equal in expectation. Both these F-statistics will
therefore correctly reflect weak-instruments problems for the 2SLS
estimator in these designs.

The robust first-stage F-statistic is given by
\begin{eqnarray*}
\hat{F}_{r} & = & \frac{1}{G}\widehat{\pi}^{\prime}\left(V\widehat{a}r_{r}\left(\widehat{\pi}\right)\right)^{-1}\widehat{\pi}\\
 & = & \frac{1}{G}\sum_{g=1}^{G}\frac{n_{g}\overline{x}_{g}^{2}}{\widehat{\sigma}_{v_{2},g}^{2}}=\frac{1}{G}\sum_{g=1}^{G}\hat{F}_{g}.
\end{eqnarray*}
It is therefore clear, that if $\hat{F}_{r}$ is large, then at least
one of the $\hat{F}_{g}$ is large. For the GMMf estimator we have
that 
\begin{eqnarray}
\widehat{\beta}_{gmmf} & = & \frac{\sum_{g=1}^{G}n_{g}\overline{x}_{g}\overline{y}_{g}/\widehat{\sigma}_{v_{2},g}^{2}}{\sum_{g=1}^{G}n_{g}\overline{x}_{g}^{2}/\widehat{\sigma}_{v_{2},g}^{2}}=\sum_{g=1}^{G}\frac{n_{g}\overline{x}_{g}^{2}/\widehat{\sigma}_{v_{2},g}^{2}}{\sum_{s=1}^{G}n_{s}\overline{x}_{s}^{2}/\widehat{\sigma}_{v_{2},s}^{2}}\widehat{\beta}_{g}\label{GMMv}\\
 & = & \sum_{g=1}^{G}\hat{w}_{gmmf,g}\widehat{\beta}_{g},\nonumber 
\end{eqnarray}
with
\[
\hat{w}_{gmmf,g}=\frac{\hat{F}_{g}}{\sum_{s=1}^{G}\hat{F}_{s}}=\frac{\hat{F}_{g}}{G\hat{F}_{r}},
\]
hence the groups with the larger F-statistics get the larger weights,
independent of the values of $\hat{\sigma}_{v_{2},s}^{2}$, mimicking
the 2SLS weights under homoskedasticity of the first-stage errors.

\subsection{Some Monte Carlo Results}

\label{subsec:MC}

We consider here the two heteroskedastic designs of \citet{IAndrewsREStat2018}
with $G=10$ groups, $\beta=0$ and one with moderate and one with
high endogeneity. Tables 9 and 12 in the Supplementary Appendix C.3
of \citet{IAndrewsREStat2018} present the values of the conditional
group-specific variance matrices $\Sigma_{g}$ as defined in (\ref{eq:Sigmas})
and the first-stage parameters, denoted $\pi_{0g}$, for $g=1,\ldots,10$.
The correlation between $u_{i}$ and $v_{2i}$ is $-0.59$ in the
moderate and we set it equal to $0.99$ in the high endogeneity case.
We multiply the first-stage parameters $\pi_{0}$ by $0.040$ and
$0.026$, such that the value of the robust $\hat{F}_{r}$ is just
over $80$ on average for $10,000$ replications and sample size $n=10,000$
in both designs. The group sizes are equal in expectation with $P\left(z_{i}=e_{g}\right)=0.1$
for all $g$. The first two rows in each panel of Table \ref{tab:grinf}
present the values of $\pi_{g}$ and $\sigma_{v_{2},g}^{2}$ for $g=1,\ldots,10$.

Tables \ref{tab:testmod} and \ref{tab:resmod} presents the weak-instruments
test results and estimation results. We use and present the critical
values based on the OLS bias benchmark, $\text{BM}_{LS}\left(\hat{W}_{\Omega_{n}},\hat{\Sigma}_{v}\right)$,
and set $\tau=0.10$. The means of the effective F-statistics are
small in both designs, and the $\hat{F}_{\text{eff}}$-based test
does not reject the null of weak instruments for the 2SLS estimator
in any of the replications. This is reflected in the bias of the 2SLS
estimator and the 2SLS-based Wald test for $H_{0}:\beta=0$ overrejects.
The means and standard deviations of the non-robust F-statistics are
virtually the same as those of the effective F-statistics, confirming
the results derived in Section \ref{sec:Grouped-Data-IV-Model}. The
means of the robust F-statistics are large, and the $\hat{F}_{r}$-based
test for weak instruments rejects the null of weak instruments in
all replications. The GMMf estimator is virtually unbiased and the
GMMf-based Wald test has good size properties. The means of the critical
values for $\hat{F}_{\text{eff}}$ are $17.09$ and $17.12$, whereas
those for $\hat{F}_{r}$ are $13.45$ in the moderate and $12.26$
in the high endogeneity design. In comparison, the 10\% relative bias
Stock and Yogo critical value for the homoskedastic case is here given
by $11.46$, see \citet{SkeelsWindmeijer2018}.

\begin{table}[t]
\caption{\label{tab:testmod} Weak-instruments test results}

\begin{centering}
\begin{tabular}{c|ccccccc}
\hline 
 & $\hat{F}$ & $\hat{F}_{\text{eff}}$ & $cv_{\text{eff}}^{LS}$ & RF & $\hat{F}_{r}$ & $cv_{r}^{LS}$ & RF\tabularnewline
\hline 
Mod Endog & 1.411 & 1.411 & 17.09 & 0 & 80.23 & 13.45 & 1\tabularnewline
 & {\small (0.738)} & {\small (0.739)} & {\small (0.110)} &  & {\small (7.025)} & {\small (0.290)} & \tabularnewline
\hline 
High Endog & 0.993 & 0.993 & 17.12 & 0 & 80.12 & 12.26 & 1\tabularnewline
 & {\small (0.599)} & {\small (0.599)} & {\small (0.113)} &  & {\small (7.081)} & {\small (0.010)} & \tabularnewline
\hline 
\end{tabular}
\par\end{centering}
\centering{}{\small Notes: Means and (st.dev.), of $10,000$ replications.
$n=10,000$, $\tau=0.1$, Rej.freq. (RF) at 5\% level.}{\small\par}
\end{table}

\begin{table}[t]
\caption{\label{tab:resmod} Estimation results}

\begin{centering}
\begin{tabular}{c|ccccc}
\hline 
 & $\widehat{\beta}_{ols}$ & $\widehat{\beta}_{2sls}$ & $\widehat{\beta}_{gmmf}$ & $Wald_{2sls}$ & $Wald_{gmmf}$\tabularnewline
\hline 
Mod Endog & -0.608 & -0.424 & -0.001 & 0.534 & 0.049\tabularnewline
 & {\small (0.011)} & {\small (0.257)} & {\small (0.563)} &  & \tabularnewline
\hline 
High Endog & 0.747 & 0.742 & 0.007 & 0.999 & 0.065\tabularnewline
 & {\small (0.001)} & {\small (0.057)} & {\small (0.029)} &  & \tabularnewline
\hline 
\end{tabular}
\par\end{centering}
\centering{}{\small Notes: Bias and (st.dev.). Rej.freq. of robust
Wald tests at 5\% level.}{\small\par}
\end{table}

The details as given in Table \ref{tab:grinf} below make clear what
is happening. It reports the population values of $\pi_{g}$, $\sigma_{v_{2},g}^{2}$,
$\mu_{n,g}^{2}=1000\pi_{g}^{2}/\sigma_{v_{2},g}^{2}$ and the mean
values of $\hat{F}_{g}$, $\hat{w}_{2sls,g}$ and $\hat{w}_{gmmf,g}=\hat{F}_{g}/\sum_{s=1}^{G}\hat{F}_{g}$.
For the moderate endogeneity design identification in the first group
is strong, with an average value of $\hat{F}_{1}=789.5$. Identification
in all other 9 groups is very weak, with the largest average value
for $\hat{F}_{5}=2.23$. The signal for group 1, $\pi_{1}^{2}$, is
somewhat larger than those for the other groups, but the population
value $\mu_{n,1}^{2}$ is large mainly due to the relatively very
small value of $\sigma_{v_{2},1}^{2}$. As detailed in (\ref{eq:w2sls}),
the 2SLS weights ignore the $\sigma_{v_{2},1}^{2}$ part of the information
in group 1 which leads to the low average value of $\hat{w}_{2sls,1}=0.127$.
This shows that the 2SLS estimator does not utilize the identification
strength of the first group well, with some larger weights given to
higher variance, but lower concentration-parameter groups.

Table \ref{tab:grinf} further shows that for the GMMf estimator almost
all weight is given to the first group, with the average of $\hat{w}_{gmmf,1}$
equal to $0.984$, resulting in the good behaviour of the GMMf estimator
in terms of bias and Wald test size. In this case the standard deviation
of the GMMf estimator is quite large relative to that of the 2SLS
estimator. This is driven by the value of $\sigma_{u,1}^{2}$, which
in this design is equal to $1.10$, much larger than $\sigma_{v_{2},1}^{2}$.
Reducing the value of $\sigma_{u,1}^{2}$ (and the value for $\sigma_{uv_{2},1}$
accordingly to keep the same correlation structure within group 1),
will reduce the standard deviation of the GMMf estimator.

The pattern of group information for the high endogeneity case is
similar to that of the moderate endogeneity case, with one informative
group, $g=10$, with an average value of $\hat{F}_{10}=792.2$. However,
the variance $\sigma_{v_{2},10}^{2}$ is now so small in relative
terms, that the 2SLS weight for group 10 has an average value of only
$\hat{w}_{2sls,10}=0.003$. The GMMf estimator corrects this, with
the average value of $\hat{w}_{gmmf,10}=0.989$. The standard deviation
of the GMMf estimates, $0.029$, is in this case smaller than that
of the 2SLS estimates, $0.057$.

\begin{table}[t]
\caption{\label{tab:grinf}Group information and estimator weights}

\begin{centering}
{\footnotesize{}%
\begin{tabular}{llrrrrrrrrrr}
\hline 
 & {\footnotesize$g$} & {\footnotesize 1} & {\footnotesize 2} & {\footnotesize 3} & {\footnotesize 4} & {\footnotesize 5} & {\footnotesize 6} & {\footnotesize 7} & {\footnotesize 8} & {\footnotesize 9} & {\footnotesize 10}\tabularnewline
\hline 
{\footnotesize ME} & {\footnotesize$\pi_{g}$} & \multicolumn{1}{r}{{\footnotesize 0.058}} & \multicolumn{1}{r}{{\footnotesize -0.023}} & \multicolumn{1}{r}{{\footnotesize 0.049}} & \multicolumn{1}{r}{{\footnotesize 0.015}} & \multicolumn{1}{r}{{\footnotesize 0.022}} & \multicolumn{1}{r}{{\footnotesize 0.008}} & \multicolumn{1}{r}{{\footnotesize -0.017}} & \multicolumn{1}{r}{{\footnotesize 0.011}} & \multicolumn{1}{r}{{\footnotesize -0.036}} & \multicolumn{1}{r}{{\footnotesize -0.040}}\tabularnewline
 & {\footnotesize$\sigma_{v_{2},g}^{2}$} & \multicolumn{1}{r}{{\footnotesize 0.004}} & \multicolumn{1}{r}{{\footnotesize 2.789}} & \multicolumn{1}{r}{{\footnotesize 4.264}} & \multicolumn{1}{r}{{\footnotesize 0.779}} & \multicolumn{1}{r}{{\footnotesize 0.395}} & \multicolumn{1}{r}{{\footnotesize 7.026}} & \multicolumn{1}{r}{{\footnotesize 1.226}} & \multicolumn{1}{r}{{\footnotesize 0.308}} & \multicolumn{1}{r}{{\footnotesize 1.709}} & \multicolumn{1}{r}{{\footnotesize 6.099}}\tabularnewline
 &  & \multicolumn{1}{r}{} & \multicolumn{1}{r}{} & \multicolumn{1}{r}{} & \multicolumn{1}{r}{} & \multicolumn{1}{r}{} & \multicolumn{1}{r}{} & \multicolumn{1}{r}{} & \multicolumn{1}{r}{} & \multicolumn{1}{r}{} & \multicolumn{1}{r}{}\tabularnewline
 & {\footnotesize$\mu_{n,g}^{2}$} & \multicolumn{1}{r}{{\footnotesize 785.7}} & \multicolumn{1}{r}{{\footnotesize 0.184}} & \multicolumn{1}{r}{{\footnotesize 0.556}} & \multicolumn{1}{r}{{\footnotesize 0.284}} & \multicolumn{1}{r}{{\footnotesize 1.190}} & \multicolumn{1}{r}{{\footnotesize 0.009}} & \multicolumn{1}{r}{{\footnotesize 0.236}} & \multicolumn{1}{r}{{\footnotesize 0.387}} & \multicolumn{1}{r}{{\footnotesize 0.770}} & \multicolumn{1}{r}{{\footnotesize 0.266}}\tabularnewline
 & {\footnotesize$\hat{F}_{g}$} & \multicolumn{1}{r}{{\footnotesize 789.5}} & \multicolumn{1}{r}{{\footnotesize 1.170}} & \multicolumn{1}{r}{{\footnotesize 1.564}} & \multicolumn{1}{r}{{\footnotesize 1.279}} & \multicolumn{1}{r}{{\footnotesize 2.225}} & \multicolumn{1}{r}{{\footnotesize 0.997}} & \multicolumn{1}{r}{{\footnotesize 1.203}} & \multicolumn{1}{r}{{\footnotesize 1.372}} & \multicolumn{1}{r}{{\footnotesize 1.798}} & \multicolumn{1}{r}{{\footnotesize 1.246}}\tabularnewline
 &  &  &  &  &  &  &  &  &  &  & \tabularnewline
 & {\footnotesize$\hat{w}_{2sls,g}$} & \multicolumn{1}{r}{{\footnotesize 0.126}} & \multicolumn{1}{r}{{\footnotesize 0.098}} & \multicolumn{1}{r}{{\footnotesize 0.178}} & \multicolumn{1}{r}{{\footnotesize 0.035}} & \multicolumn{1}{r}{{\footnotesize 0.031}} & \multicolumn{1}{r}{{\footnotesize 0.180}} & \multicolumn{1}{r}{{\footnotesize 0.049}} & \multicolumn{1}{r}{{\footnotesize 0.015}} & \multicolumn{1}{r}{{\footnotesize 0.096}} & \multicolumn{1}{r}{{\footnotesize 0.192}}\tabularnewline
 & {\footnotesize$\hat{w}_{gmmf,g}$} & {\footnotesize 0.984} & {\footnotesize 0.002} & {\footnotesize 0.002} & {\footnotesize 0.002} & {\footnotesize 0.003} & {\footnotesize 0.001} & {\footnotesize 0.002} & {\footnotesize 0.002} & {\footnotesize 0.002} & {\footnotesize 0.002}\tabularnewline
\hline 
{\footnotesize HE} & {\footnotesize$100\cdot\pi_{g}$} & {\footnotesize -0.021} & {\footnotesize 0.095} & {\footnotesize -0.484} & {\footnotesize -0.069} & {\footnotesize 0.159} & {\footnotesize -0.028} & {\footnotesize 0.101} & {\footnotesize -0.418} & {\footnotesize 0.450} & {\footnotesize -0.546}\tabularnewline
 & {\footnotesize$\sigma_{v_{2},g}^{2}$} & {\footnotesize 1.600} & {\footnotesize 0.478} & {\footnotesize 2.975} & {\footnotesize 1.142} & {\footnotesize 0.174} & {\footnotesize 0.145} & {\footnotesize 4.658} & {\footnotesize 1.963} & {\footnotesize 2.990} & {\footnotesize 0.38$\cdot a$}\tabularnewline
 &  &  &  &  &  &  &  &  &  &  & \tabularnewline
 & {\footnotesize$\mu_{n,g}^{2}$} & {\footnotesize 0.28$\cdot a$} & {\footnotesize 0.002} & {\footnotesize 0.008} & {\footnotesize 4.2$\cdot a$} & {\footnotesize 0.015} & {\footnotesize 5.6$\cdot a$} & {\footnotesize 2.2$\cdot a$} & {\footnotesize 0.009} & {\footnotesize 0.007} & {\footnotesize 789.9}\tabularnewline
 & {\footnotesize$\hat{F}_{g}$} & {\footnotesize 0.998} & {\footnotesize 1.017} & {\footnotesize 0.979} & {\footnotesize 1.010} & {\footnotesize 1.034} & {\footnotesize 0.984} & {\footnotesize 0.977} & {\footnotesize 1.031} & {\footnotesize 0.997} & {\footnotesize 792.2}\tabularnewline
 &  &  &  &  &  &  &  &  &  &  & \tabularnewline
 & {\footnotesize$\hat{w}_{2sls,g}$} & {\footnotesize 0.111} & {\footnotesize 0.040} & {\footnotesize 0.177} & {\footnotesize 0.085} & {\footnotesize 0.016} & {\footnotesize 0.013} & {\footnotesize 0.242} & {\footnotesize 0.134} & {\footnotesize 0.181} & {\footnotesize 0.003}\tabularnewline
 & {\footnotesize$\hat{w}_{gmmf,g}$} & \multicolumn{1}{r}{{\footnotesize 0.001}} & \multicolumn{1}{r}{{\footnotesize 0.001}} & \multicolumn{1}{r}{{\footnotesize 0.001}} & \multicolumn{1}{r}{{\footnotesize 0.001}} & \multicolumn{1}{r}{{\footnotesize 0.001}} & \multicolumn{1}{r}{{\footnotesize 0.001}} & \multicolumn{1}{r}{{\footnotesize 0.001}} & \multicolumn{1}{r}{{\footnotesize 0.001}} & \multicolumn{1}{r}{{\footnotesize 0.001}} & \multicolumn{1}{r}{{\footnotesize 0.989}}\tabularnewline
\hline 
\end{tabular}}{\footnotesize\par}
\par\end{centering}
\centering{}{\small Notes: }$\mu_{n,g}^{2}=1000\pi_{g}^{2}/\sigma_{v_{2},g}^{2}$;
$a=10^{-4}$
\end{table}

\begin{figure}[t]
\begin{centering}
\includegraphics[scale=0.55]{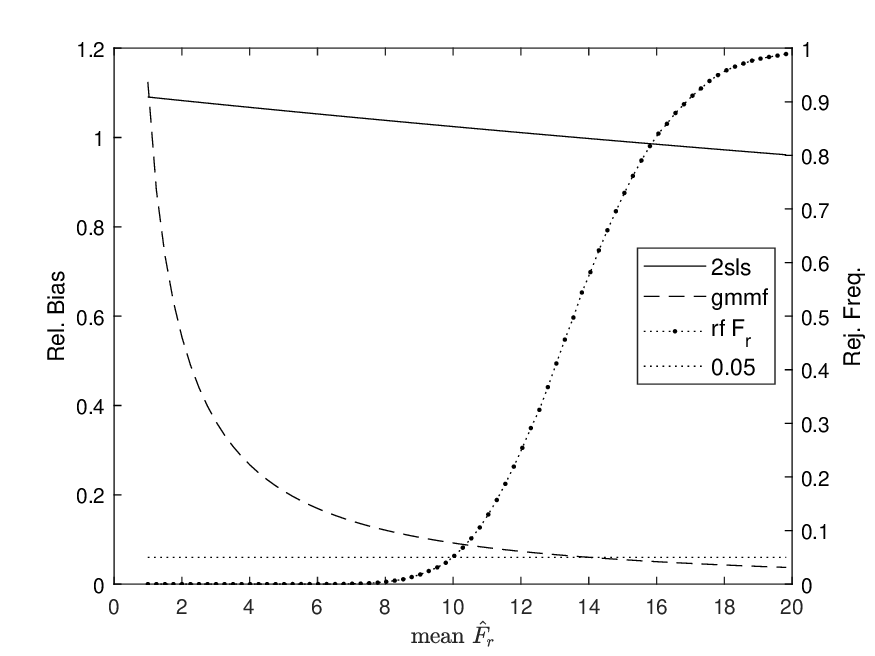}\includegraphics[scale=0.55]{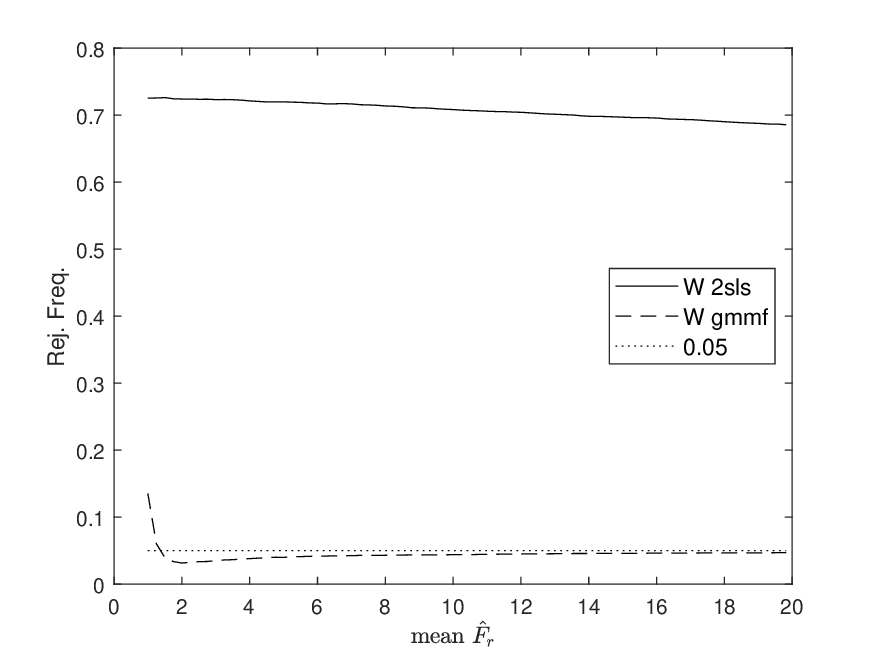}
\par\end{centering}
\begin{centering}
\includegraphics[scale=0.55]{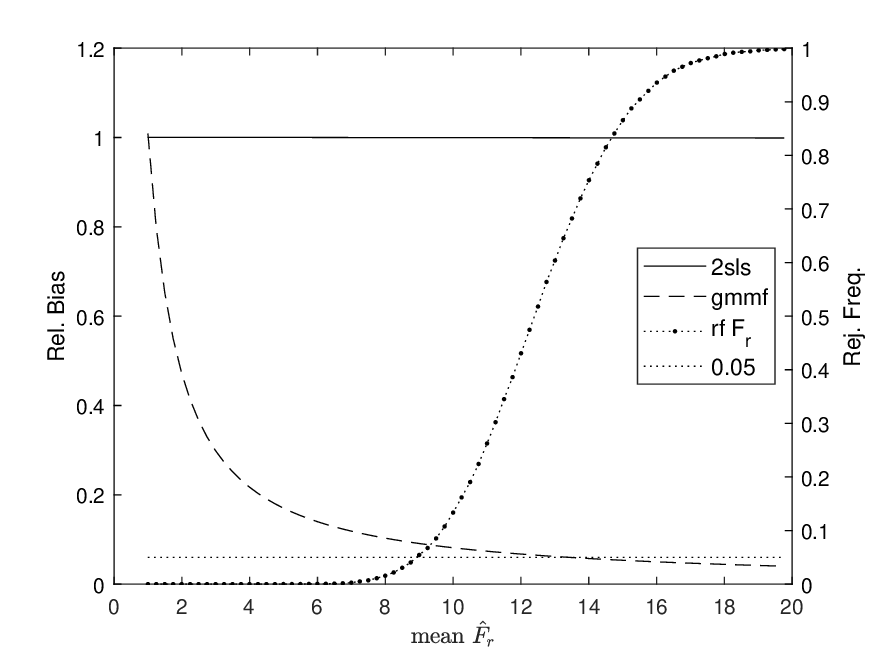}\includegraphics[scale=0.55]{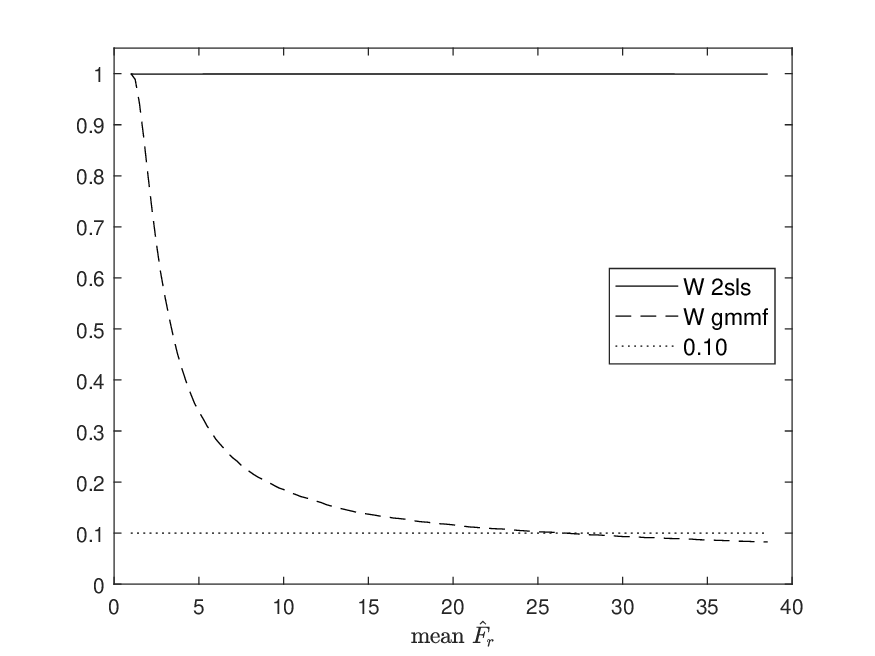}
\par\end{centering}
\caption{\label{fig:rejfrelb} Heteroskedastic design. Top: Moderate Endogeneity.
Bottom: High Endogeneity. Left: Bias of 2SLS and GMMf estimators relative
to OLS bias, and $\hat{F}_{r}$-based weak-instrument test rejection
frequencies, $\alpha=0.05$, $\tau=0.10$, least-squares benchmark
bias. The ``0.05'' line refers to the rej. freq. Right: Rejection
frequencies of robust Wald tests, $\alpha=0.05$.}
\end{figure}

The left panels of Figure \ref{fig:rejfrelb} displays the relative
bias of the 2SLS and GMMf estimators, relative to that of the OLS
estimator, as a function of the mean values of the robust F-statistic
$\hat{F}_{r}$, together with the rejection frequency of the $\hat{F}_{r}$-based
test for weak instruments, using the critical values from the least-squares
benchmark bias. We present the relative bias here to be in line with
the homoskedastic case as presented below. Different values of $\hat{F}_{r}$
are obtained by different values of the scalar $e$ when setting the
first-stage parameters $\pi=e\pi_{0}$. The relative bias of the GMMf
estimator decreases quite rapidly with increasing values of $\hat{F}_{r}$.
For the moderate endogeneity case, the test has a rejection frequency
of $5\%$ at a mean $\hat{F}_{r}$ of $10.03$, with the relative
bias of the GMMf estimator at that point equal to $0.092$. As shown
in the top right-hand panel of Figure \ref{fig:rejfrelb}, the GMMf
estimator based Wald test is well behaved in terms of size, with hardly
any size distortion for mean values of $\hat{F}_{r}$ larger than
5. The GMMf relative bias picture for the high-endogeneity case is
very similar to that of the moderate-endogeneity case. Here the $\hat{F}_{r}$
based test for weak instruments has a rejection frequency of $5\%$
at a mean $\hat{F}_{r}$ of $8.98$, with the relative bias there
being $0.091$. As for the homoskedastic case, where the Wald test
size deviation from nominal size is larger for larger values of $\rho_{uv_{2}}^{2}$,
the GMMf Wald test has a worse size performance in the high-endogeneity
design, and has a $10\%$ rejection frequency at a mean $\hat{F}_{r}$
of $26.64$. This would imply a critical value at the $5\%$ level
of around $32$, which compares to the Stock and Yogo weak-instruments
critical value of $38.54$ for a Wald test size of $10\%$ at the
$5\%$ nominal level.

\subsubsection{Homoskedastic Design}

We next consider the homoskedastic design for the moderate endogeneity
case with $\Sigma_{uv_{2}}=\frac{1}{G}\sum_{g=1}^{G}\Sigma_{uv_{2},g}$,
resulting in
\[
\Sigma_{uv_{2}}=\left[\begin{array}{cc}
2.57 & -1.50\\
-1.50 & 2.46
\end{array}\right],
\]
with $\rho_{uv_{2}}=-0.59$, as above. We consider smaller sample
sizes of $n=250$ and $n=500$, or group sizes of $25$ or $50$ on
average, to compare the weak-instrument finite sample behaviour of
the GMMf estimator to that of the 2SLS estimator. In particular, the
noise induced by estimation of $W_{2}^{-1}$ may adversely affect
the GMMf estimator.

The results in Figure \ref{fig:Hom} shows that for this design and
sample sizes the relative biases and Wald rejection frequencies are
virtually identical for the two estimators, with the standard deviations
of the GMMf estimates slightly larger than those of the 2SLS estimator,
as expected. The rejection frequencies of the $\hat{F}_{r}$-based
test are here closer to those of the standard Stock and Yogo $\hat{F}$-based
test compared to the rejection frequencies of the $\hat{F}_{\text{eff}}$-based
test, with the latter test more conservative.

\begin{figure}[t]
\begin{centering}
\includegraphics[scale=0.55]{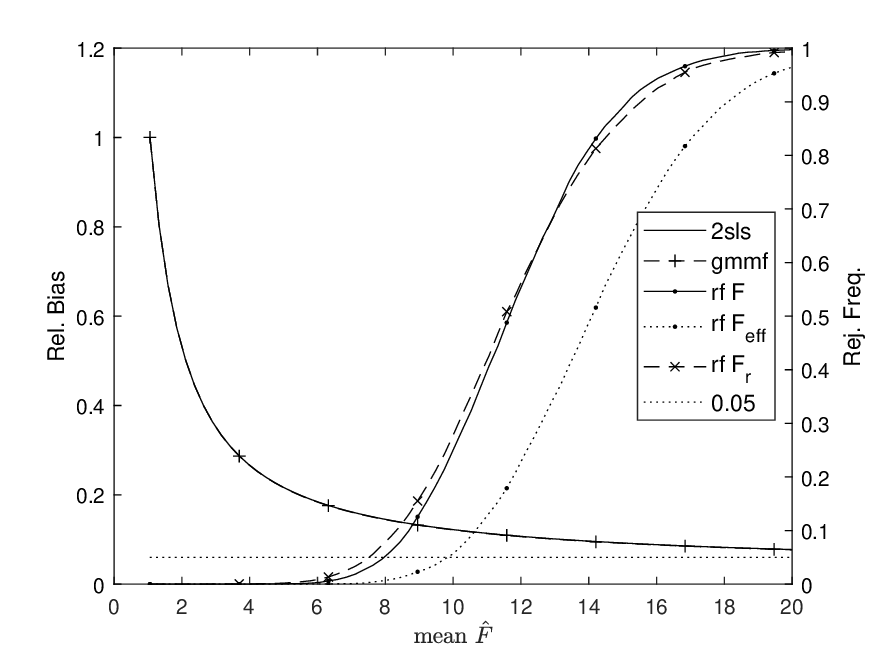}\includegraphics[scale=0.55]{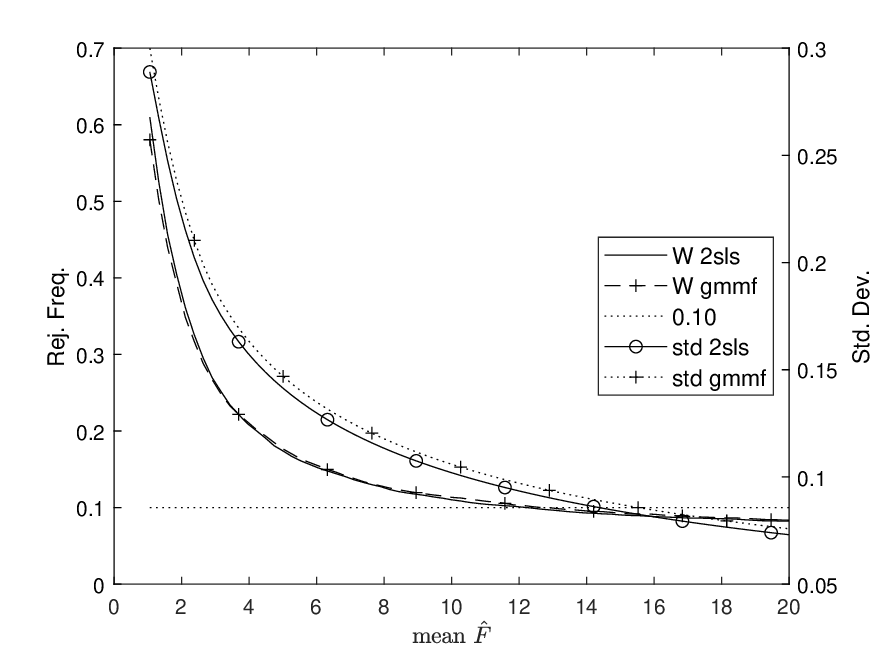}
\par\end{centering}
\begin{centering}
5\includegraphics[scale=0.55]{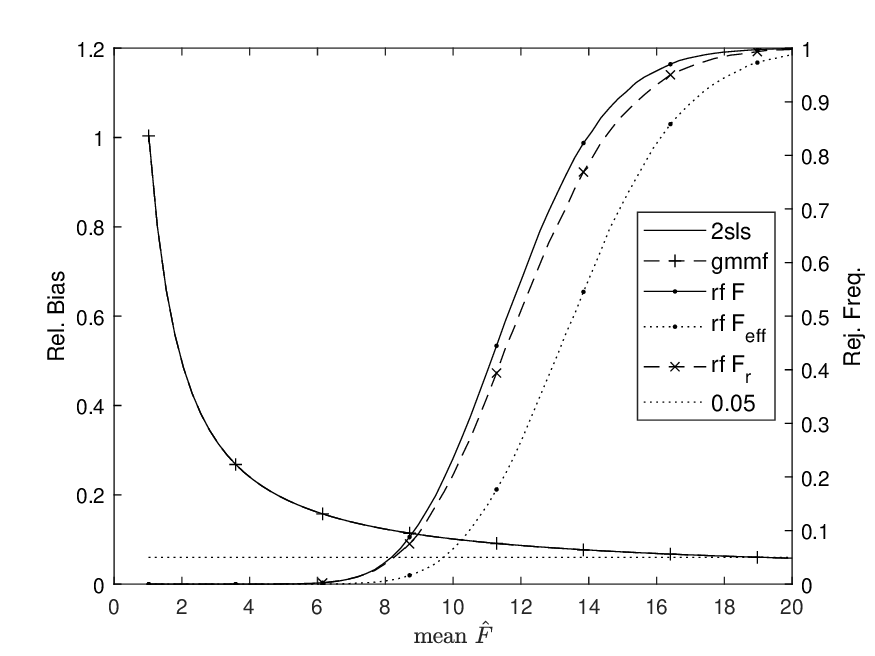}\includegraphics[scale=0.55]{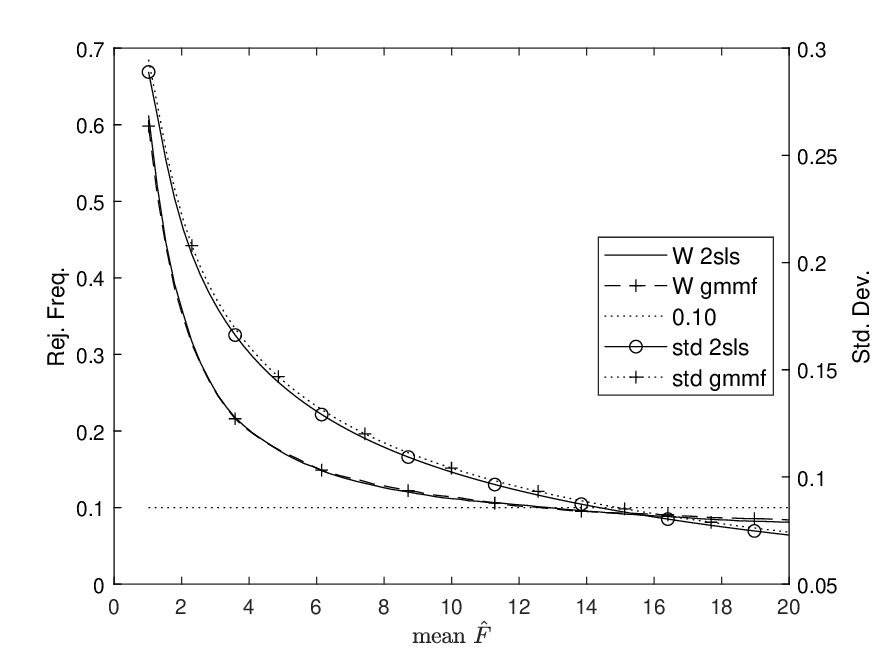}
\par\end{centering}
\caption{\label{fig:Hom} Homoskedastic design, moderate endogeneity. Top:
$n=250$. Bottom: $n=500$. Left: Relative bias and weak-instrument
tests rejection frequencies. Right: Wald test rejection frequencies
and standard deviations of 2SLS and GMMf estimates. }
\end{figure}

\subsection{Relative Values of $\hat{F}_{\text{eff}}$ and $\hat{F}_{r}$ and
Biases of 2SLS and GMMf Estimators}

\label{subsec:Relative-Values}

From the results on the relative magnitude of the concentration parameters
as discussed in Section \ref{sec:RobF} and footnote \ref{fn:concpar},
we can change the parameter values in the grouped-data design of Section
\ref{subsec:MC} such that $\hat{F}_{\text{eff}}>\hat{F}_{r}$ in
expectation. The design is described in Appendix \ref{subsec:Fefflarge}
where Table \ref{tab:testmodlarge} presents Monte Carlo results for
the weak-instruments tests and estimation results. The null of weak
instruments is rejected in all replications for the 2SLS estimator,
whereas it is not rejected in virtually all replications for the GMMf
estimator and the bias of the 2SLS estimator is smaller than that
of the GMMf estimator.

As per the results of Tables \ref{tab:testmod} and \ref{tab:resmod},
a large value of $\hat{F}_{r}\gg cv_{r}^{LS}$ indicates that the
GMMf estimator is quite well behaved in terms of bias. But when that
is the case, a value of $\hat{F}_{\text{eff}}<cv_{\text{eff}}^{LS}$,
and so not rejecting the null of weak instruments for the 2SLS estimator,
does not necessarily imply that the 2SLS bias is larger than the GMMf
bias (and vice versa). Consider an example with fixed group sizes
$\frac{n_{g}}{n}=f_{g}$, $0<f_{g}<1$, for $g=1,\ldots,G$. Then
for 2SLS, 
\[
\mu_{2sls}^{2}\coloneqq\mu_{Q_{zz}^{-1}}^{2}=\frac{\sum_{g=1}^{G}c_{g}^{2}f_{g}}{\sum_{g=1}^{G}\sigma_{v_{2},g}^{2}},
\]
and for GMMf, 
\[
\mu_{gmmf}^{2}\coloneqq\mu_{W_{2}^{-1}}^{2}=\frac{1}{G}\sum_{g=1}^{G}c_{g}^{2}f_{g}/\sigma_{v_{2},g}^{2}.
\]
Then the Nagar bias (\ref{eq:NagBias}) for the 2SLS and GMMf estimators
are respectively given by
\[
N_{2sls}=\frac{\sum_{g=1}^{G}\left(1-2\frac{c_{g}^{2}f_{g}}{\sum_{s=1}^{G}c_{s}^{2}f_{s}}\right)\sigma_{uv_{2},g}}{\sum_{g=1}^{G}c_{g}^{2}f_{g}}
\]
\[
N_{gmmf}=\frac{\sum_{g=1}^{G}\left(1-2\frac{c_{g}^{2}f_{g}/\sigma_{v_{2},g}^{2}}{\sum_{s=1}^{G}c_{s}^{2}f_{s}/\sigma_{v_{2},s}^{2}}\right)\frac{\sigma_{uv_{2},g}}{\sigma_{v_{2},g}^{2}}}{\sum_{g=1}^{G}c_{g}^{2}f_{g}/\sigma_{v_{2},g}^{2}}.
\]
One can find therefore parameter values in this design, such that
$N_{2sls}<N_{gmmf}$ whereas $\mu_{2sls}<\mu_{gmmf}$. It follows
for example, that if $c_{g}^{2}f_{g}=c_{f}$, $c_{f}>0$ for $g=1,\ldots,G$,
then $N_{2sls}=0$ if $\sum_{g=1}^{G}\sigma_{uv_{2},g}=0$, irrespective
of the value of $\mu_{2sls}$. In that case, there is no overall endogeneity,
as $\rho_{uv_{2}}=0$. (Conversely, in the constant $c_{g}^{2}f_{g}$
case, $N_{gmmf}=0$ if $\sum_{g=1}^{G}\frac{\sigma_{uv_{2},g}}{\sigma_{v_{2},g}^{4}}=0$
irrespective of the value of $\mu_{gmmf}$). 

As another example, consider again $G=10$, with $f_{g}=0.1$ for
$g=1,\ldots,10$, with parameter values $c_{g}$ and group variance
matrices $\Sigma_{uv_{2},g}$ given in Appendix \ref{subsec:Parameter-Values}.
The resulting overall endogeneity is a moderate $\rho_{uv_{2}}=0.244$.
We get a value of $\mu_{gmmf}^{2}=43.09$ and a smaller value of $\mu_{2sls}^{2}=8.45$,
whereas the Nagar bias approximations are virtually identical at $N_{2sls}=0.022<N_{gmmf}=0.023$.
Table \ref{tab:zero2slstest} in Appendix \ref{subsec:Parameter-Values}
presents some Monte Carlo results for this design, for $n=10,000$,
confirming the theoretical results. We find there that the null of
weak instruments is not rejected in virtually all replications for
the 2SLS estimator, whereas it is rejected in all replications for
the GMMf estimator. But the Monte Carlo bias of the 2SLS estimator,
$0.022$, is slightly smaller than that of the of the GMMf estimator,
$0.024$, as indicated by the Nagar bias approximations.

From these latter examples it is clear that a situation with $\hat{F}_{r}>cv_{r}^{LS}$
and $\hat{F}_{\text{eff}}<cv_{\text{eff}}^{LS}$ does not necessarily
imply that the 2SLS estimator is more biased than the GMMf estimator
(and vice versa). However, randomly sampling values $c_{g}$ and group
variance matrices $\Sigma_{uv_{2},g}$, for $g=1,\ldots,10$, and
collecting the first $1000$ of those sets where $\left|\rho_{uv_{2}}\right|>0.2$,
$5<\mu_{2sls}<10$ and $40<\mu_{gmmf}<45$, we find for $989$, or
$99\%$ of those, the Nagar bias of the 2SLS estimator to be larger
than that of the GMMf estimator. 

\section{Considerations for Practice}

\label{sec:Practice}

The \citet{IAndrewsREStat2018} grouped-data IV designs are quite
extreme in the variation of $\sigma_{v_{2},g}^{2}$, leading to the
large differences between the values of $\hat{F}_{\text{eff}}$ and
$\hat{F}_{r}$ and between the performances of the 2SLS and GMMf estimators.
Note that these results carry over to a model with a constant and
a full set of mutually exclusive binary indicators as instruments,
when the variances $\sigma_{v_{2},g}^{2}$ for at least two groups
are relatively small and their $\pi_{g}$ coefficients are different.
This is the case if we for example change $\sigma_{v_{2},2}^{2}$
in the moderate endogeneity design above to be equal to the small
$\sigma_{v_{2},1}^{2}$. An example where this could be relevant is
the judge fixed effects design as pioneered by \citet{Kling2006}.
There are many papers using judge effects as instrumental variables,
Table 1 in \citet{Frandsen2023} presents a subset of these. \citet{Stevenson2018}
studied the effect of pretrial detention on conviction, using judge
indicators as instruments, as cases are randomly assigned to judges.
As the treatment is here binary, with variance $\pi\left(1-\pi\right)$,
a very lenient (small $\pi$) and a very strict judge (large $\pi$)
in terms of sending defendants to pretrial detention have small values
of $\sigma_{v_{2}}^{2}$, but clearly different values of $\pi$.
Unlike the 2SLS estimator, the GMMf estimator takes the differential
strengths of the instruments due to the different values of $\sigma_{v_{2}}^{2}$
into account, giving more weight to very lenient and very strict judges.\footnote{I would like to thank an anonymous referee for this example.}

For any single-endogenous variable application, most packages will
compute the robust F-statistic for estimation results with robust
standard errors. For example, in Stata, \citet{Stata}, the robust
first-stage F-statistic is provided with the output of ``ivregress''
or ``ivreg2'', \citet*{ivreg2}, whereas ``weakivtest'', \citet{weakivtest},
calculates the effective F-statistic and critical values for the weak-instruments
test. An extended version of the latter, called ``gfweakivtest''\footnote{Available from https://github.com/stat0357/gfweakivtest.}
also calculates the robust F-statistic and its weak-instruments critical
values. It further includes the critical values based on the least-squares
benchmark bias for both $\hat{F}_{r}$ and $\hat{F}_{\text{eff}}$,
and presents the estimation results for the GMMf estimator. As per
the advice of \citet{AndrewsStockSun2019}, 2SLS estimation results
with robust standard errors should be accompanied by the effective
F-statistic and its critical value. If the situation is as in the
\citet{IAndrewsREStat2018} examples above, then the GMMf estimator
could be the preferred estimator. 

\citet{Windmeijer2023} presents such comparisons for a study with
a set of mutually exclusive binary indicators as instruments, and
one of the \textit{American Economic Review} studies as considered
in the review paper by \citet{AndrewsStockSun2019}. This is from
\citet{Stephens2014} who study the effect of schooling on wages,
using data from the 1960-1980 US Censuses of Population. The endogenous
variable is years of schooling for individual $i$, born in state
$s$ in year $t$, and the instruments are three indicator variables
$RS7$, $RS8$ and $RS9$, corresponding to being required to attend
seven, eight or nine or more years of schooling, respectively. All
specifications include state-of-birth and year-of-birth fixed effects,
and the computed standard errors are robust to heteroskedasticity
and clustering at the state-of-birth/year-of-birth cell. \citet{Stephens2014}
report the robust first-stage F-statistics $\hat{F}_{r}$ in their
Table 1, which presents eight sets of estimates of the returns of
schooling on log weekly wages for four different samples and two different
model specifications for each sample. None of the $\hat{F}_{r}$ statistics
indicate an underidentification problem. But no effective F-statistics
were reported. As the estimator used is the 2SLS estimator, it is
therefore important to consider whether the $\hat{F}_{r}$ statistic
misrepresents weak-instruments bias of the 2SLS estimator, in the
sense that a large value of $\hat{F}_{r}$ may not be an indicator
of a good performance of the 2SLS estimator. Table 3 in \citet{Windmeijer2023}
shows that this is not the case here. The $\hat{F}_{\text{eff}}$
and $\hat{F}_{r}$ based tests for weak instruments both reject and
don't reject the null of weak instruments for the same specifications,
their values are similar in magnitude and the 2SLS and GMMf estimation
results are virtually identical when the null of weak instruments
is rejected and the Hansen $J$-test does not indicate misspecification.
This is a reassuring result for the 2SLS estimates that were accompanied
by larger values of $\hat{F}_{r}$, with the only two cases where
the null of weak instruments was not rejected had values of $\hat{F}_{r}$
equal to 8.22 and 6.34 with those of $\hat{F}_{\text{eff}}$ equal
to 8.11 and 6.13 respectively.

\section{Concluding Remarks}

\label{sec:Concluding-Remarks}

For models with a single endogenous explanatory variable, we have
introduced a class of generalized effective F-statistics as defined
in (\ref{eq:geffF}) in relation to a class of linear GMM estimators
given in (\ref{eq:bgmm}) and have shown that the \citet{OleaPflueger2013}
weak-instruments testing procedure that they established for the effective
F-statistic in relation to the Nagar bias of the 2SLS estimator applies
to this extended class. In particular, the standard robust F-statistic
is a member of this class and is associated with the behaviour in
terms of Nagar bias of the GMMf estimator, which has its weight matrix
based on the first-stage residuals. We then focused on a comparison
of the effective F-statistic and the robust F-statistic and the associated
weak-instrument behaviours of the 2SLS and GMMf estimators. In particular,
we have shown that and explained why the GMMf estimator's performance
is much better in terms of bias than that of the 2SLS estimator in
the grouped-data designs of \citet{IAndrewsREStat2018}, where the
robust F-statistic can take very large values, but the effective F-statistic
is very small. One should therefore in general not use the robust
F-statistic to gauge instrument strength in relation to the performance
of the 2SLS estimator, \citet[pp 738-739]{AndrewsStockSun2019}, but
as shown here, it can be used as a weak-instruments test in relation
to the Nagar bias of the GMMf estimator. In practice, therefore, both
the effective F-statistic and robust F-statistic should be reported,
together with their critical values, and the GMMf estimator could
be considered in cases where there is a clear discrepancy with a large
value for the robust F-statistic rejecting the null of weak instruments,
and when the effective F-statistic is small and does not reject its
null of weak instruments.

We have not focused here on the wider applicability of the class of
generalized effective F-statistics and their associated GMM estimators,
but an example is the one-step \citet{ABond1991} GMM estimator for
panel data models with a single endogenous variable. Two-step estimators
do not fall in the class because of the presence of estimated structural
parameters in the weight matrix, but one could test for weak instruments
in this setting, fixing the parameter of the endogenous variable in
the weight matrix, for example under a specific null value of interest.

A topic for future research for the general heteroskedasticity setting
is an extension to the linear model with more than one endogenous
variable. \citet{LewisMertens2022} is an extension of the \citet{OleaPflueger2013}
method to the multiple endogenous variable case for the 2SLS estimator,
but they do not consider such an extension for the wider class of
GMM estimators. Future research should also address the weak-instruments
Wald size properties for both the single and multiple endogenous variables
settings.

\section*{Acknowledgments}

I would like to thank the editor Xiaohong Chen, an associate editor
and two anonymous referees for their detailed comments and suggestions
that helped to improve the paper substantially. Further thanks to
Benedikt Koch and seminar participants for helpful comments.

\section*{Appendix}

\global\long\def\thesection{A}%
\global\long\def\theequation{A.\arabic{equation}}%
\global\long\def\thetable{A\arabic{table}}%
\setcounter{table}{0}\setcounter{equation}{0}

\subsection{A Grouped-Data Design with $\hat{F}_{\text{eff}}>\hat{F}_{r}$}

\label{subsec:Fefflarge}

From the results on the relative magnitude of the concentration parameters
as discussed in Section \ref{sec:RobF} and footnote \ref{fn:concpar},
we can change the parameter values such that $\hat{F}_{\text{eff}}>\hat{F}_{r}$
in expectation in the grouped-data design. Here we take the moderate-endogeneity
design of \citet{IAndrewsREStat2018}, but change the value of $\sigma_{v_{2},1}^{2}$
from $0.004$ to $20$, adjusting $\sigma_{u,1}^{2}$ and $\sigma_{uv_{2},1}$
accordingly, and change the value of $\pi_{1}$ from $0.058$ to $1.414$.
This results in the expected group specific concentration parameter
$\mu_{n,1}^{2}=100$, and so the first group is again the informative
group, but has now a relatively large variance $\sigma_{v_{2}}^{2}$
and a relatively large signal $\pi_{1}^{2}$. The estimation results
are given in Table \ref{tab:testmodlarge}. $\hat{F}_{\text{eff}}$
indicates that there is no weak-instruments problem for 2SLS, whereas
$\hat{F}_{r}$ shows that there is for the GMMf estimator. The 2SLS
estimator is virtually unbiased and the 2SLS Wald test is well behaved,
whereas the GMMf estimator displays bias with the its Wald test overrejecting.
In this case, the 2SLS estimator gives a large weight of $0.98$ to
group 1, whereas this weight is $0.89$ for GMMf.

\begin{table}[th]
\caption{\label{tab:testmodlarge} Weak-instruments tests and estimation results,
$\sigma_{v_{2},1}^{2}=20$, $\pi_{1}=1.414$}

\begin{centering}
\begin{tabular}{ccccccc}
\hline 
$\hat{F}$ & $\hat{F}_{\text{eff}}$ & $cv_{\text{eff}}^{LS}$ & RF & $\hat{F}_{r}$ & $cv_{r}^{LS}$ & RF\tabularnewline
\hline 
46.10 & 46.12 & 18.55 & 1 & 11.42 & 16.03 & 0.025\tabularnewline
{\small (9.11)} & {\small (9.17)} & {\small (0.116)} &  & {\small (2.16)} & {\small (0.191)} & \tabularnewline
\hline 
$\widehat{\beta}_{ols}$ & $\widehat{\beta}_{2sls}$ & $\widehat{\beta}_{gmmf}$ & $Wald_{2sls}$ & $Wald_{gmmf}$ &  & \tabularnewline
-0.316 & -0.007 & -0.065 & 0.056 & 0.217 &  & \tabularnewline
{\small (0.011)} & {\small (0.024)} & {\small (0.061)} &  &  &  & \tabularnewline
\hline 
\end{tabular}
\par\end{centering}
\centering{}{\small Notes: See notes to Tables \ref{tab:testmod} and
\ref{tab:resmod}. Moderate Endogeneity.}{\small\par}
\end{table}

\subsection{Parameter Values for Design of Section \ref{subsec:Relative-Values}}

\label{subsec:Parameter-Values}

The parameter values $c_{g}$ and $\Sigma_{uv_{2}}=\left[\begin{array}{cc}
\sigma_{u,g}^{2} & \sigma_{uv_{2},g}\\
\sigma_{uv_{2},g} & \sigma_{v_{2},g}^{2}
\end{array}\right]$ are given by
\begin{table}[H]
\caption{Parameter values.}

\begin{centering}
\begin{tabular}{crccc}
\hline 
$g$ & $c_{g}$ & $\sigma_{u,g}^{2}$ & $\sigma_{uv_{2},g}$ & $\sigma_{v_{2},g}^{2}$\tabularnewline
\hline 
1 & 20.6393 & 9.0052 & 1.7135 & 4.2487\tabularnewline
2 & 27.6284 & 3.4060 & 1.7847 & 9.9668\tabularnewline
3 & -3.3019 & 2.3741 & 2.8222 & 6.0015\tabularnewline
4 & -38.7569 & 1.7522 & -0.7409 & 0.4370\tabularnewline
5 & -11.1463 & 3.5420 & -2.4995 & 8.6788\tabularnewline
6 & 18.2092 & 3.2771 & 3.0059 & 4.0456\tabularnewline
7 & -0.4646 & 0.0538 & 0.3084 & 6.9979\tabularnewline
8 & 25.0219 & 6.2319 & 4.8593 & 8.2675\tabularnewline
9 & -25.6606 & 5.8019 & -0.4336 & 4.2698\tabularnewline
10 & 5.9592 & 7.3973 & 0.8086 & 0.0968\tabularnewline
\hline 
\end{tabular}
\par\end{centering}
\end{table}

Table \ref{tab:zero2slstest} presents some Monte Carlo results for
this design, for $n=10,000$, confirming the theoretical results.
We find here that the null of weak instruments is not rejected in
virtually all replications for the 2SLS estimator, whereas it is rejected
in all replications for the GMMf estimator. But the Monte Carlo bias
of the 2SLS estimator, $0.022$, is slightly smaller than that of
the of the GMMf estimator, $0.024$, as indicated by the Nagar bias
approximations as detailed in Section \ref{subsec:Relative-Values}.

\begin{table}[th]
\caption{\label{tab:zero2slstest} Weak-instruments tests and estimation results}

\begin{centering}
\begin{tabular}{cccccc}
\hline 
$\hat{F}_{\text{eff}}$ & $cv_{\text{eff}}^{LS}$ & RF & $\hat{F}_{r}$ & $cv_{r}^{LS}$ & RF\tabularnewline
\hline 
9.49 & 15.85 & 0.00 & 44.24 & 19.47 & 1\tabularnewline
{\small (1.83)} & {\small (0.10)} &  & {\small (4.45)} & {\small (0.15)} & \tabularnewline
\hline 
$\widehat{\beta}_{ols}$ & $\widehat{\beta}_{2sls}$ & $\widehat{\beta}_{gmmf}$ & $Wald_{2sls}$ & $Wald_{gmmf}$ & \tabularnewline
\hline 
0.218 & 0.022 & 0.024 & 0.062 & 0.049 & \tabularnewline
{\small (0.009)} & {\small (0.092)} & {\small (0.149)} &  &  & \tabularnewline
\hline 
\end{tabular}
\par\end{centering}
\centering{}{\small Notes: See notes to Tables \ref{tab:testmod} and
\ref{tab:resmod}. Design as in text and Appendix.}{\small\par}
\end{table}

\subsection{Proof of Lemma 1}

\label{subsec:Proof-of-Lemma 1}

It follows from the first-stage and reduced-form model specifications
(\ref{eq:fstage}) and (\ref{eq:reduced}) and Assumptions \ref{Ass:all}
and \ref{Ass:Om} that, as $n\rightarrow\infty$,
\begin{align}
\frac{1}{\sqrt{n}}\left(\begin{array}{c}
\Omega_{n}^{1/2}Z'y\\
\Omega_{n}^{1/2}Z'x
\end{array}\right) & =\left(\begin{array}{c}
\Omega_{n}^{1/2}\left(\frac{1}{n}Z'Z\right)c\beta+\frac{1}{\sqrt{n}}\Omega_{n}^{1/2}Z'v_{1}\\
\Omega_{n}^{1/2}\left(\frac{1}{n}Z'Z\right)c+\frac{1}{\sqrt{n}}\Omega_{n}^{1/2}Z'v_{2}
\end{array}\right)\nonumber \\
 & \overset{d}{\rightarrow}\left(\begin{array}{c}
\gamma_{\Omega,1}\\
\gamma_{\Omega,2}
\end{array}\right):=\left(\begin{array}{c}
c_{\Omega}\beta+\Omega^{1/2}\psi_{1}\\
c_{\Omega}+\Omega^{1/2}\psi_{2}
\end{array}\right)\sim N\left(\left(\begin{array}{c}
c_{\Omega}\beta\\
c_{\Omega}
\end{array}\right),W_{\Omega}\right),\label{eq:gamom}
\end{align}
where $c_{\Omega}=\Omega^{1/2}Q_{zz}c$.

As
\[
\hat{\beta}_{\Omega_{n}}=\frac{x'Z\Omega_{n}Z'y}{x'Z\Omega_{n}Z'x},
\]
it follows that 
\[
\hat{\beta}_{\Omega_{n}}-\beta=\frac{x'Z\Omega_{n}Z'u}{x'Z\Omega_{n}Z'x}=\frac{x'Z\Omega_{n}Z'\left(v_{1}-\beta v_{2}\right)}{x'Z\Omega_{n}Z'x}.
\]
Under Assumptions \ref{Ass:all} and \ref{Ass:Om} it follows from
(\ref{eq:gamom}) and the continuous mapping theorem that
\[
\hat{\beta}_{\Omega_{n}}-\beta\overset{d}{\rightarrow}\beta_{\Omega}^{*}=\left(\gamma_{\Omega,2}^{\prime}\gamma_{\Omega,2}\right)^{-1}\gamma_{\Omega,2}^{\prime}\left(\gamma_{\Omega,1}-\beta\gamma_{\Omega,2}\right).
\]

As
\[
\hat{F}_{\text{geff}}\left(\Omega_{n}\right)=\frac{x'Z\Omega_{n}Z'x}{n\text{tr}\left(\hat{W}_{2}\Omega_{n}\right)}
\]
it follows Assumptions \ref{Ass:all} and \ref{Ass:Om} that
\[
\hat{F}_{\text{geff}}\left(\Omega_{n}\right)\overset{d}{\rightarrow}\frac{\gamma_{\Omega,2}^{\prime}\gamma_{\Omega,2}}{\text{tr}\left(W_{2}\Omega\right)}
\]
again from (\ref{eq:gamom}) and the continuous mapping theorem.

\subsection{Proof of Nagar Bias Result of Theorem 1}

\label{subsec:Proof-of-Nagar Bias}

The Nagar bias result is obtained as follows. Let $S\left(\beta,W_{\Omega}\right)$
be as defined in (\ref{eq:SbW}), and
\begin{align*}
\xi & =S_{1}^{-1/2}\left(\beta,W_{\Omega}\right)\left(\gamma_{\Omega,1}-\beta\gamma_{\Omega,2}\right)\sim\mathcal{N}\left(0,1\right),\\
\nu & =W_{\Omega,2}^{-1/2}\left(\gamma_{\Omega,2}-c_{\Omega}\right)\sim\mathcal{N}\left(0,1\right).
\end{align*}
 We then have that
\begin{align*}
\beta_{\Omega}^{*} & =\frac{\gamma_{\Omega,2}^{\prime}\left(\gamma_{\Omega,1}-\beta\gamma_{\Omega,2}\right)}{\gamma_{\Omega,2}^{\prime}\gamma_{\Omega,2}}\\
 & =\frac{c_{\Omega}^{\prime}S_{1}^{1/2}\left(\beta,W_{\Omega}\right)\xi+\nu'W_{\Omega,2}^{1/2}S_{1}^{1/2}\left(\beta,W_{\Omega}\right)\xi}{c_{\Omega}^{\prime}c_{\Omega}+2c_{\Omega}'W_{\Omega,2}^{1/2}S_{1}^{1/2}\nu+\nu'W_{\Omega,2}\nu}.
\end{align*}
It follows that
\[
\Vert c_{\Omega}\Vert\beta_{\Omega}^{*}=\frac{c_{\Omega,0}^{\prime}S_{1}^{1/2}\left(\beta,W_{\Omega}\right)\xi+\frac{\nu'W_{\Omega,2}^{1/2}S_{1}^{1/2}\left(\beta,W_{\Omega}\right)\xi}{\sqrt{\text{tr}\left(W_{\Omega,2}\right)}}\mu_{\Omega}^{-1}}{1+\frac{2c_{\Omega,0}'W_{\Omega,2}^{1/2}S_{1}^{1/2}\nu}{\sqrt{\text{tr}\left(W_{\Omega,2}\right)}}\mu_{\Omega}^{-1}+\frac{\nu'W_{\Omega,2}\nu}{\text{tr}\left(W_{\Omega,2}\right)}\mu_{\Omega}^{-2}},
\]
where $\mu_{\Omega}^{2}=\Vert c_{\Omega}\Vert^{2}/\text{tr}\left(W_{\Omega,2}\right)$.
Then from \citet[(6.2)]{Rothenberg1984}, we get the second-order
Edgeworth, \citet{Nagar1959} approximation
\begin{align*}
\mathbb{E}\left(\beta_{\Omega}^{*}\right) & \approx\frac{1}{\mu_{\Omega}^{2}}\frac{1}{\text{tr}\left(W_{\Omega,2}\right)}\mathbb{E}\left[\nu'W_{\Omega,2}^{1/2}S_{1}^{1/2}\left(\beta,W_{\Omega}\right)\xi-2c_{\Omega,0}^{\prime}S_{1}^{1/2}\left(\beta,W_{\Omega}\right)\xi c_{\Omega,0}'W_{\Omega,2}^{1/2}\nu\right]\\
 & =\frac{1}{\mu_{\Omega}^{2}}\left(\frac{\text{tr}\left(S_{12}\left(\beta,W_{\Omega}\right)-2c_{\Omega,0}'S_{12}\left(\beta,W_{\Omega}\right)c_{\Omega,0}\right)}{\text{tr}\left(W_{\Omega,2}\right)}\right).
\end{align*}

\bibliographystyle{ecta}
\bibliography{RobF}

\end{document}